\ifdefined\XeTeXversion\else\pdfoutput=1\fi
\documentclass[11pt]{article}

\usepackage[letterpaper,margin=1in]{geometry}
\usepackage{amsmath,amssymb,amsthm,mathtools}
\usepackage{thmtools}
\usepackage{microtype}
\usepackage{enumitem}
\usepackage{booktabs}
\usepackage{graphicx}
\usepackage[font=small,labelfont=bf]{caption}
\usepackage{array}
\usepackage[hidelinks]{hyperref}
\hypersetup{
  pdftitle={Associativity of Multiplication Is Hard for Resolution},
  pdfauthor={Vincent Liew},
  pdfsubject={An exponential lower bound on the size of general resolution refutations of multiplier associativity},
  pdfkeywords={proof complexity, resolution, SAT, multiplication, associativity, lower bound},
}

\setlist[enumerate]{leftmargin=2.2em,itemsep=0.25em,topsep=0.35em}
\setlist[itemize]{leftmargin=1.8em,itemsep=0.25em,topsep=0.35em}
\allowdisplaybreaks

\newtheorem{theorem}{Theorem}[section]
\newtheorem{lemma}[theorem]{Lemma}
\newtheorem{proposition}[theorem]{Proposition}

\newtheorem{observation}[theorem]{Observation}
\newtheorem*{informalmain}{Theorem~\ref{thm:main} (informal)}
\newtheorem*{informalfrege}{Theorem~\ref{thm:frege} (informal)}
\theoremstyle{definition}
\newtheorem{definition}[theorem]{Definition}
\newtheorem*{definition*}{Definition}
\theoremstyle{remark}
\newtheorem{remark}[theorem]{Remark}

\newcommand{\Assoc}{\mathsf{Assoc}}
\newcommand{\Mul}{\mathsf{Mul}}
\newcommand{\PM}{\mathsf{PM}}
\newcommand{\SPM}{\mathsf{SPM}}
\newcommand{\Grid}{\mathsf{Grid}}
\newcommand{\sizeRes}{\operatorname{size}_{\mathrm{Res}}}
\newcommand{\sizeF}[1]{\operatorname{size}_{\mathrm{Frege}_{#1}}}
\newcommand{\supp}{\operatorname{supp}}
\newcommand{\MAJ}{\operatorname{MAJ}}
\newcommand{\N}{\mathbb{N}}
\newcommand{\col}{\operatorname{col}}
\newcommand{\Col}{\operatorname{Col}}
\newcommand{\vv}[1]{\mathbf{#1}}

\title{Associativity of Multiplication Is Hard for Resolution}
\author{Vincent Liew\\Category Labs\\\texttt{vliew@vliew.com}}
\date{\today}

\begin{document}
\pagenumbering{roman}
\maketitle

\begin{abstract}
SAT solvers are empirically known to perform poorly when reasoning about multiplication. Yet for over a decade we have lacked a theoretical explanation for this phenomenon. CDCL SAT solvers implicitly search for resolution proofs, and no lower bound on proof size has ruled out the existence of short proofs that solvers simply fail to find.

We give the first lower bound of this kind by showing that general resolution proofs of the associativity of $n$-bit multiplication require size $2^{\Omega((n/\log n)^{1/4})}$. This lower bound holds for a broad class of multiplier
encodings based on partial product summation, including the standard array and Wallace-tree multipliers used to bit-blast multiplication
in SMT solvers. This result resolves an open problem of Beame and Liew.

The proof constructs a reduction from a perfect-matching principle on bounded-degree bipartite expander graphs to multiplier associativity. Itsykson, Slabodkin, and Sokolov proved that this principle is hard for resolution. The lower bound
for multiplier associativity follows.

The same reduction, when combined with H\r{a}stad's recent lower bound for the perfect-matching principle of the odd grid, yields an exponential lower bound for multiplier associativity in the stronger bounded-depth Frege proof systems.
\end{abstract}
\thispagestyle{empty}
\newpage
\pagenumbering{arabic}

\section{Introduction}
\label{sec:intro}

Why do SAT solvers struggle with multiplication?
Despite dramatic leaps in efficiency in recent decades, the inability of SAT solvers to reason about integer multiplication has remained a major obstacle for many applications of SAT and SMT solving, among them bounded model checking~\cite{bbkmt26}, equivalence checking of arithmetic hardware~\cite{kbk-fmcad,bkks25}, and the verification of cryptographic code and smart contracts~\cite{cryptoline,intblast,polysat}. The persistence of this difficulty is illustrated by a recent evaluation on multiplier equivalence-checking benchmarks submitted to the 2025 SAT Competition~\cite{bkks25}, where the modern SAT solver \textsc{Kissat}~\cite{kissat} did not outperform the twenty-year-old solver \textsc{MiniSat}~\cite{minisat}.

This weakness is felt most acutely in bit-vector solving, where SAT solvers serve as the core reasoning engine. Bit-vector solvers ultimately handle arithmetic by \emph{bit-blasting}: each bit becomes a Boolean variable and each operation is represented by a Boolean circuit, so a bit-vector formula involving multiplication reduces to a SAT instance containing a multiplier circuit. While bit-blasting is typically effective for bit-vector formulas containing, for example, addition, comparisons, and bitwise operations, performance collapses on formulas requiring substantial reasoning about multiplication~\cite{liew-thesis,bbkmt26}.

Practitioners have responded by engineering around this weakness, turning to forms of reasoning better suited to multiplication: multiplier verification via Gr\"obner bases~\cite{sayed-ahmed,rbk,kbk-fmcad,revsca,kbbn}; word-level algebraic decision procedures such as PolySAT~\cite{polysat}, int-blasting~\cite{intblast}, and cvc5's finite-field solver~\cite{okt}; as well as rewriting and abstraction layers designed to keep multiplication away from the underlying SAT solver~\cite{npz}. The reliance on rewriting runs especially deep~\cite{nprzbt}: before bit-blasting, the rewrite engines of all major bit-vector solvers (e.g., Bitwuzla~\cite{bitwuzla}, cvc5~\cite{cvc5}, and Z3~\cite{z3}) normalize arithmetic terms modulo commutativity and associativity of multiplication, so that the underlying SAT solver never has to reason about these properties.

Proof complexity provides the standard framework for explaining this kind of empirical
hardness. In order for a CDCL SAT solver to determine that a formula is unsatisfiable,
it must implicitly construct a \emph{resolution refutation} of
that formula~\cite{bks}. The running time of a CDCL solver is therefore bounded
below by the size of the smallest resolution refutation. A classic example is the pigeonhole principle: Haken's
exponential lower bound~\cite{haken} explains why state-of-the-art solvers
fail to show that $n+1$ pigeons do not fit into $n$ holes, for modest
values of $n$. In this spirit, Biere conjectured in 2016 that even basic algebraic
properties of multiplication circuits, such as commutativity, require
exponentially large resolution refutations~\cite{biere-fields}.

Beame and Liew gave a surprising negative answer to this conjecture for the case of degree two ring identities, which include commutativity and distributivity, but not associativity~\cite{beame-liew}. But although they constructed polynomial size proofs for these identities, the concrete size of these proofs was still large. Their best upper bounds, for commutativity (i.e., $\vv{x}\vv{y}=\vv{y}\vv{x}$) and distributivity (i.e., $\vv{x}(\vv{y}+\vv{z})=\vv{x}\vv{y} + \vv{x}\vv{z}$), scaled as $O(n^6 \log n)$, where $n$ is the bit-width. Follow-up experiments indicated that CDCL SAT solvers do not find these proofs on standard multiplier encodings, even with substantial guidance~\cite{liew-thesis,bbkmt26}. A different set of experiments in Liew's thesis suggested that even if SAT solvers did find these proofs, their size would still be prohibitive in practice~\cite{liew-thesis}. So although Beame and Liew's result resolved Biere's specific conjecture, it gave limited insight into the broader aim of explaining \emph{why} SAT solvers consistently struggle with multiplication.

Meanwhile, the practical workarounds described above are incomplete, and when they fail, every major solver still falls back to bit-blasting. Moreover, this retreat from CDCL has rested entirely on empirical grounds. Without proven lower bounds, the possibility has remained that these workarounds are unnecessary. Perhaps resolution proofs of practical size exist for these problems, and current CDCL solvers simply lack the heuristics to find them.

For associativity of multiplication, we show that such small proofs do not exist. In particular, we prove an exponential general resolution lower bound for multiplier associativity, resolving a problem left open by Beame and Liew~\cite{beame-liew}. This is the first resolution lower bound for a natural formula involving multiplication that holds without restricting the order in which variables are resolved. Prior to this result, no such lower bound was known even for tree-like or regular resolution.

\begin{informalmain}
Let $\Assoc_n$ be the CNF formula asserting $(\vv{x}\vv{y})\vv{z}\ne\vv{x}(\vv{y}\vv{z})$ for $n$-bit input bit-vectors $\vv{x},\vv{y},\vv{z}$, with multiplication encoded by an array or Wallace-tree multiplier. Then resolution refutations of $\Assoc_n$ require size
$
2^{\Omega((n/\log n)^{1/4})}.
$
\end{informalmain}

This lower bound is not specific to these two multiplier architectures. It applies to a general class of \emph{strip-local} multiplier encodings, defined in Section~\ref{sec:strip-local}, which captures most of the standard multipliers used for bit-blasting.

We prove the lower bound by constructing a reduction from the perfect-matching principle on bounded-degree bipartite graphs. The reduction reveals a perfect matching instance embedded inside the associativity formula: we fix $\vv{x}$ and $\vv{z}$ to sparse bit-vectors whose nonzero positions represent the vertices of the graph, and leave selected bits of $\vv{y}$ free to serve as edge variables. We use the bipartite expanders of Itsykson, Slabodkin, and Sokolov~\cite{iss} (see also~\cite{ioss}), who proved that this principle on these graphs is hard for resolution. In order to absorb the auxiliary circuit variables of the multiplier encodings, we extend the perfect-matching principle to a \emph{star-local} variant and show that it remains hard for resolution.

Our second result concerns the stronger bounded-depth Frege proof systems, where the lines of a proof are bounded-depth formulas rather than clauses. H\r{a}stad recently proved that the perfect-matching principle of the odd grid is hard for these systems~\cite{hastad-php}. By applying the same reduction as before, with the odd grid in place of the expander, we get the following lower bound.

\begin{informalfrege}
There is a constant $\gamma>0$ such that for all $d\le\gamma\log n/\log\log n$, depth $d$ Frege refutations of $\Assoc_n$ require size $\exp(n^{\gamma/d})$.
\end{informalfrege}

This result implies an exponential lower bound for resolution as well, but with a smaller exponent than in Theorem~\ref{thm:main}. Moreover, H\r{a}stad's lower bound for the odd grid is significantly more technical than the lower bound for bipartite expanders. For these reasons we work with expanders for the resolution lower bound.

\subsection{Related work}
The proof complexity of multipliers has been previously studied in the resolution, cutting planes, and polynomial calculus proof systems. However, multiplier associativity in particular has resisted attempts at upper and lower bounds in each of these proof systems.

\paragraph{Resolution.}
As discussed above, Beame and Liew~\cite{beame-liew} constructed
polynomial size resolution proofs for degree two identities, including commutativity and distributivity. On the lower-bound side, \cite{proofdoor} gave the first proof complexity lower bound involving
multiplier circuits, but only in the restrictive \emph{partially-ordered
resolution} proof system, and only for specific partial orders.

\paragraph{Cutting planes.}
\cite{lbden} gave
$O(n^2)$-length cutting planes proofs for a large class of degree two
identities. These proofs are optimal, as the multiplier encodings themselves consist of $\Omega(n^2)$ inequalities. They noted, however, that their method does not extend to
associativity, whose cutting planes proof complexity remains open.

\paragraph{Polynomial calculus.}
Gr\"obner-basis methods, whose derivations are polynomial calculus proofs, were
introduced for gate-level multiplier verification~\cite{sayed-ahmed,rbk}. These methods were later scaled to large, highly optimized
multipliers by Mahzoon, Gro\ss{}e, and
Drechsler~\cite{revsca}, and, with the help of an auxiliary SAT solver, by Kaufmann, Biere, and Kauers~\cite{kbk-fmcad}. The SAT solver was subsequently eliminated by encoding negated literals with dual
variables~\cite{kbbn}, so that the resulting procedure yields a
single polynomial calculus certificate.

Beyond equivalence checking, Kaufmann, Biere, and Kauers showed that Gr\"obner-basis methods
efficiently verify word-level commutativity of multipliers in practice~\cite{kbk}, and that
the underlying derivation amounts to an $O(n^2)$-length polynomial calculus
proof of word-level commutativity for array multipliers. Their idea extends to show that
word-level ring identities in general, including associativity, admit
$O(n^2)$-length polynomial calculus proofs. However, we do not know whether the corresponding bit-level identities have small proofs. \cite{lbden} also proved that polynomial calculus requires exponentially many monomials to extract the bit-level consequences of a
word-level equality. So for associativity, the word-level equality $(\vv{x}\vv{y})\vv{z} = \vv{x}(\vv{y}\vv{z})$ does not immediately yield a small polynomial calculus proof that the middle output bits agree.

This gap between word-level and bit-level reasoning is also visible in practice: PolySAT~\cite{polysat}, int-blasting~\cite{intblast}, and cvc5's finite-field solver~\cite{okt} all reason with polynomial constraints at the word level, and each adds dedicated mechanisms for recovering bit-level facts from that reasoning.

\paragraph{Organization.}
Section~\ref{sec:prelims} collects preliminaries. Section~\ref{sec:assoc-formula} defines the associativity formula $\Assoc_n$. Section~\ref{sec:strip-local} defines strip-local multiplier encodings, and shows that array and Wallace-tree multipliers are strip-local. Section~\ref{sec:hard-source} introduces the perfect-matching principle $\PM(H)$ and its star-local
extension $\SPM(H)$, and proves that the latter is hard for resolution. Section~\ref{sec:outline} outlines the main steps and
ideas of the reduction. Section~\ref{sec:rho} constructs the key ingredient of the reduction: the restriction $\rho$. This restriction reveals a perfect matching inside $\Assoc_n$. Section~\ref{sec:reduction} proves the associativity lower bound for resolution (Theorem~\ref{thm:main}) by carrying out the reduction from the star-local extension of perfect matching to the associativity formula. Section~\ref{sec:frege} carries out the same reduction in bounded-depth Frege systems, with the odd grid in place of the expander, resulting in Theorem~\ref{thm:frege}.
Section~\ref{sec:conclusion} concludes with open problems.

\section{Preliminaries}
\label{sec:prelims}

\paragraph{Notation.}
For integers $a\le b$, we write $[a,b)=\{a,a+1,\ldots,b-1\}$. A length $n$ bit-vector is written as
$\vv{x}=(x_0,\ldots,x_{n-1})\in\{0,1\}^n$, and represents the integer $\sum_{i<n}x_i2^i$. Following~\cite{beame-liew}, we refer to each multiplier by its output bit-vector, so the multiplier with inputs $\vv{x}, \vv{y}$ is labeled $\vv{x}\vv{y}$.

\paragraph{CNF formulas.}
A \emph{literal} is a Boolean variable or its negation; a \emph{clause} is a disjunction of literals, which we treat as a set. A clause containing complementary literals is a
\emph{tautology}. A CNF formula is a conjunction of clauses. The \emph{width} of a clause is the number of literals it contains, and the width $w(\Phi)$ of a CNF $\Phi$ is the maximum width of its clauses.

\paragraph{Resolution.}
The \emph{resolution rule} produces the clause $A\vee B$, called the \emph{resolvent}, from clauses $A\vee x$ and $B\vee\neg x$. A
\emph{resolution derivation} is a sequence of clauses, each of which is either an \emph{initial clause} of the input CNF $\Phi$, or follows from two prior clauses via the resolution rule. The \emph{size} of a derivation is the length of the sequence of clauses.
A derivation is a \emph{refutation} of $\Phi$ if it ends with the empty clause $\bot$. The minimum size of a resolution refutation of $\Phi$ is denoted by
$\sizeRes(\Phi)$. The \emph{width} of a derivation is the maximum width of its clauses, and $w(\Phi\vdash\bot)$ denotes the minimum width of a resolution refutation of $\Phi$. 
Resolution is sound and complete: it can only derive clauses that are \emph{implied} by $\Phi$, and conversely, if $\Phi$ implies a nontautological
clause $C$, then some subclause of $C$ has a resolution derivation from $\Phi$.

We record the size--width tradeoff of Ben-Sasson and Wigderson in the form we use.

\begin{theorem}[\cite{bw}]
\label{thm:bw}
Let $\Phi$ be an unsatisfiable CNF with $N$ variables. Then
\[
\sizeRes(\Phi)\ge \exp\!\Bigl(\Omega\Bigl(\tfrac{(w(\Phi\vdash\bot)-w(\Phi))^2}{N}\Bigr)\Bigr).
\]
In particular, if $\Phi$ has constant width and $O(m)$ variables, and every resolution refutation of $\Phi$ has width $\Omega(m)$, then
$\sizeRes(\Phi)\ge 2^{\Omega(m)}$.
\end{theorem}

\paragraph{Substitutions.}
A \emph{substitution} $\sigma$ maps each variable of a CNF $\Phi$ to a constant
or to a literal, possibly over a different variable set. We write $C^\sigma$ or $\Phi^\sigma$ for the result of applying $\sigma$ to a clause $C$, or a CNF $\Phi$, respectively. As usual, $\Phi^\sigma$ omits the clauses that have become trivially true.

\paragraph{Restrictions.}
A restriction is the special case of a substitution where every variable maps to a
constant or to itself. We denote the result of applying a restriction $\rho$ to a clause $C$ by $C\!\upharpoonright_\rho$, and the result of applying this restriction to a formula $\Phi$ as $\Phi\!\upharpoonright_\rho$.

\begin{lemma}[Substitution lemma]
\label{lem:substitution}
Let $\sigma$ be a substitution. If a CNF $\Phi$ has a resolution refutation of size $s$, then $\Phi^\sigma$ has a resolution refutation of
size at most $s$.
\end{lemma}

\begin{definition}[Exact multiplier encoding]
\label{def:exact-encoding}
An \emph{exact multiplier encoding} $\Mul_n(\vv{x},\vv{y};\vv{x}\vv{y})$ is a
CNF whose variables are the bits of the $n$-bit inputs $\vv{x}$ and $\vv{y}$ together
with auxiliary variables called \emph{wires}. Each wire computes a Boolean function of the input
bits, in the sense that the satisfying assignments of the CNF are exactly the assignments in which
every wire takes the value of its function. The wires include a designated
$2n$-bit \emph{output} $\vv{x}\vv{y}$, computing the product of the input bit-vectors.
\end{definition}

\begin{definition}[Partial products, columns]
\label{def:partial-products}
The \emph{partial products} of two length-$n$ bit-vectors $\vv{x}$ and $\vv{y}$ are
\[
t_{i,j}=x_i\wedge y_j,\qquad 0\le i,j<n.
\]
The product $\vv{x}\vv{y}$ is the weighted sum of these partial products: $\vv{x}\vv{y}=\sum_{i,j}t_{i,j}2^{i+j}$. The partial product $t_{i,j}$ has \emph{arithmetic weight} $2^{i+j}$, and we say that it lies in \emph{column} $i+j$.
\end{definition}

\paragraph{Array multiplier.} 
A full adder with inputs $a_0,a_1,a_2$ each of weight $2^k$ has outputs
\[
s=a_0\oplus a_1\oplus a_2,
\qquad
c=\MAJ(a_0,a_1,a_2),
\]
of weights $2^k$ and $2^{k+1}$, so the inputs and outputs of a full adder have the same total weight. 
The $n$-bit array multiplier, as shown in Figure~\ref{fig:array-multiplier}, computes each partial product $t_{i,j}$ by an AND gate and uses a grid of full adders to compute the sum of all the partial products. Observe that this is an exact multiplier encoding.

\begin{figure}[ht]
\centering
\includegraphics[width=0.75\textwidth]{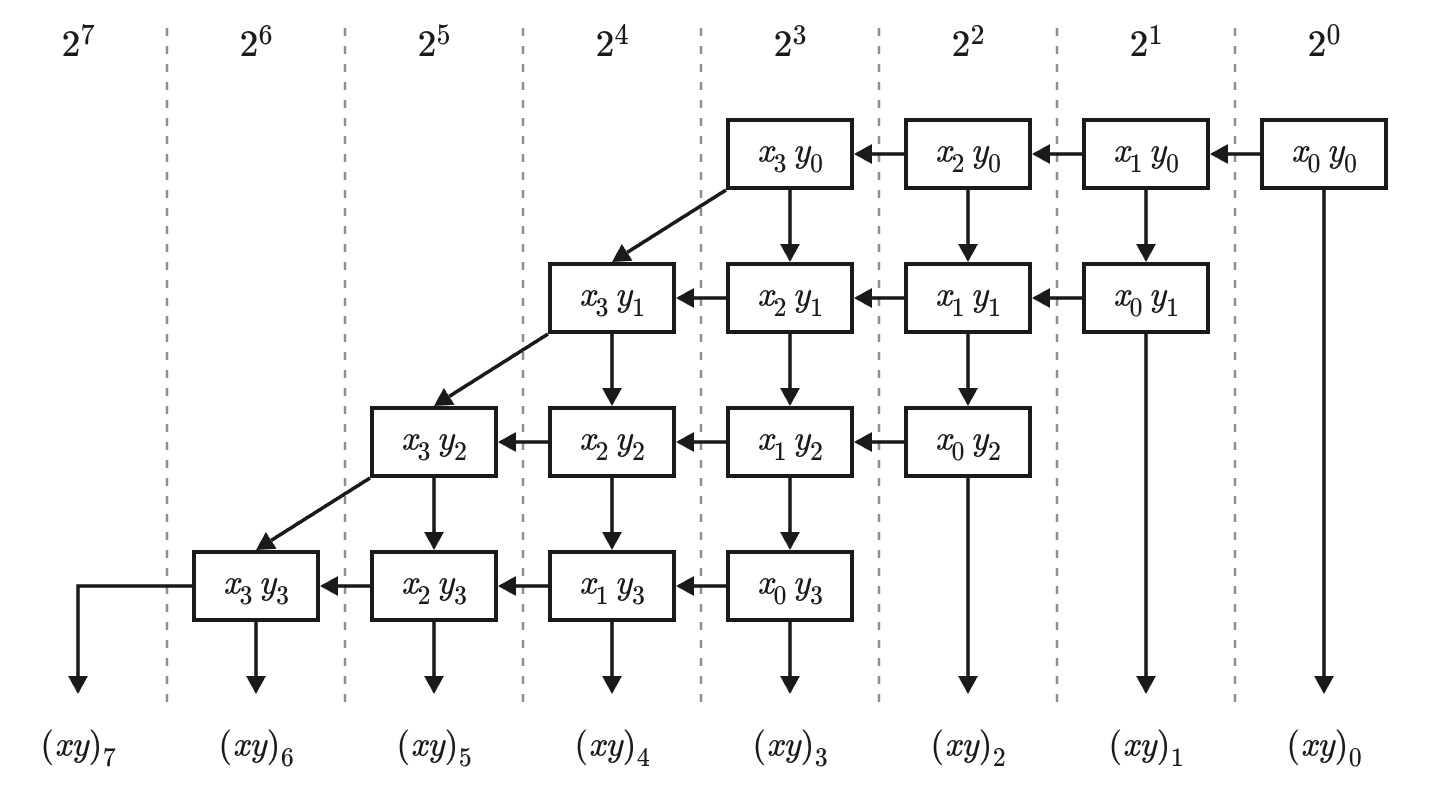}
\caption{The $4$-bit array multiplier. Each box represents a full adder, and adds the partial product $t_{i,j}=x_i\wedge y_j$ to the sum bit arriving
from the row above (vertical arrows) and the carry arriving from the previous column (horizontal and diagonal arrows).}
\label{fig:array-multiplier}
\end{figure}

\paragraph{Bipartite graphs and matchings.}
Let $G=(U,V,E)$ be a bipartite graph with parts $U$ and $V$ and edges $E$. For an edge $e\in E$ we write $e=(u_e,v_e)$ with $u_e\in U$ and $v_e\in V$. For a vertex $w$, we write $E(w)$ for the set of edges incident to $w$, called the \emph{star} of $w$;
the degree of $w$ is $|E(w)|$. A \emph{matching} is a set of pairwise disjoint edges; it \emph{covers} a vertex set if every vertex of the
set is an endpoint of an edge in the matching.

\section{The associativity formula}
\label{sec:assoc-formula}

The associativity formula $\Assoc_n$ contains a circuit of two multipliers outputting $(\vv{x}\vv{y})\vv{z}$, a separate circuit of two multipliers outputting $\vv{x}(\vv{y}\vv{z})$, and a \emph{miter circuit}~\cite{brand93} which asserts that these two outputs disagree.\footnote{The term presumably comes from the shape of the construction: two circuits converging on an XOR gate resemble a bishop's miter, as illustrated by Biere~\cite[slide~36]{biere-fields}.}

\begin{definition}[Associativity formula]
\label{def:assoc}
Let $\Mul_n$ be an exact multiplier encoding. The unsatisfiable formula $\Assoc_n$, shown in Figure~\ref{fig:assoc-miter}, consists of the four multiplier encodings
\[
\Mul_n(\vv{x},\vv{y};\vv{x}\vv{y}),\qquad \Mul_n(\vv{y},\vv{z};\vv{y}\vv{z}),
\]
\[
\Mul_{2n}(\vv{x}\vv{y},\vv{z};(\vv{x}\vv{y})\vv{z}),\qquad \Mul_{2n}(\vv{x},\vv{y}\vv{z};\vv{x}(\vv{y}\vv{z})),
\]
along with \emph{miter variables} $e_i$ and \emph{miter clauses} encoding the constraints
\[
e_i\leftrightarrow \bigl(((\vv{x}\vv{y})\vv{z})_i\oplus(\vv{x}(\vv{y}\vv{z}))_i\bigr),\qquad 0\le i<4n,
\]
\[
e_0\vee\cdots\vee e_{4n-1}.
\]
Here the outputs of the inner multipliers, $\vv{x}\vv{y}$ and $\vv{y}\vv{z}$, have $2n$ bits, the $n$-bit inputs $\vv{x}$ and $\vv{z}$ to the outer multipliers are zero-padded to $2n$ bits, and the outputs of the outer multipliers $(\vv{x}\vv{y})\vv{z}$ and $\vv{x}(\vv{y}\vv{z})$ have $4n$ bits.

We refer to $\vv{x}\vv{y}$ and $\vv{y}\vv{z}$ as the
\emph{inner} multipliers and $(\vv{x}\vv{y})\vv{z}$ and $\vv{x}(\vv{y}\vv{z})$ as the \emph{outer} multipliers.
\end{definition}

In fixed-width bit-vector arithmetic, multiplier outputs are truncated to the input bit-width. We work with untruncated multipliers for convenience, as the lower bound will apply unchanged to
truncated multipliers (see Remark~\ref{rem:truncated}).

\begin{figure}[ht]
\centering
\includegraphics[width=0.5\textwidth]{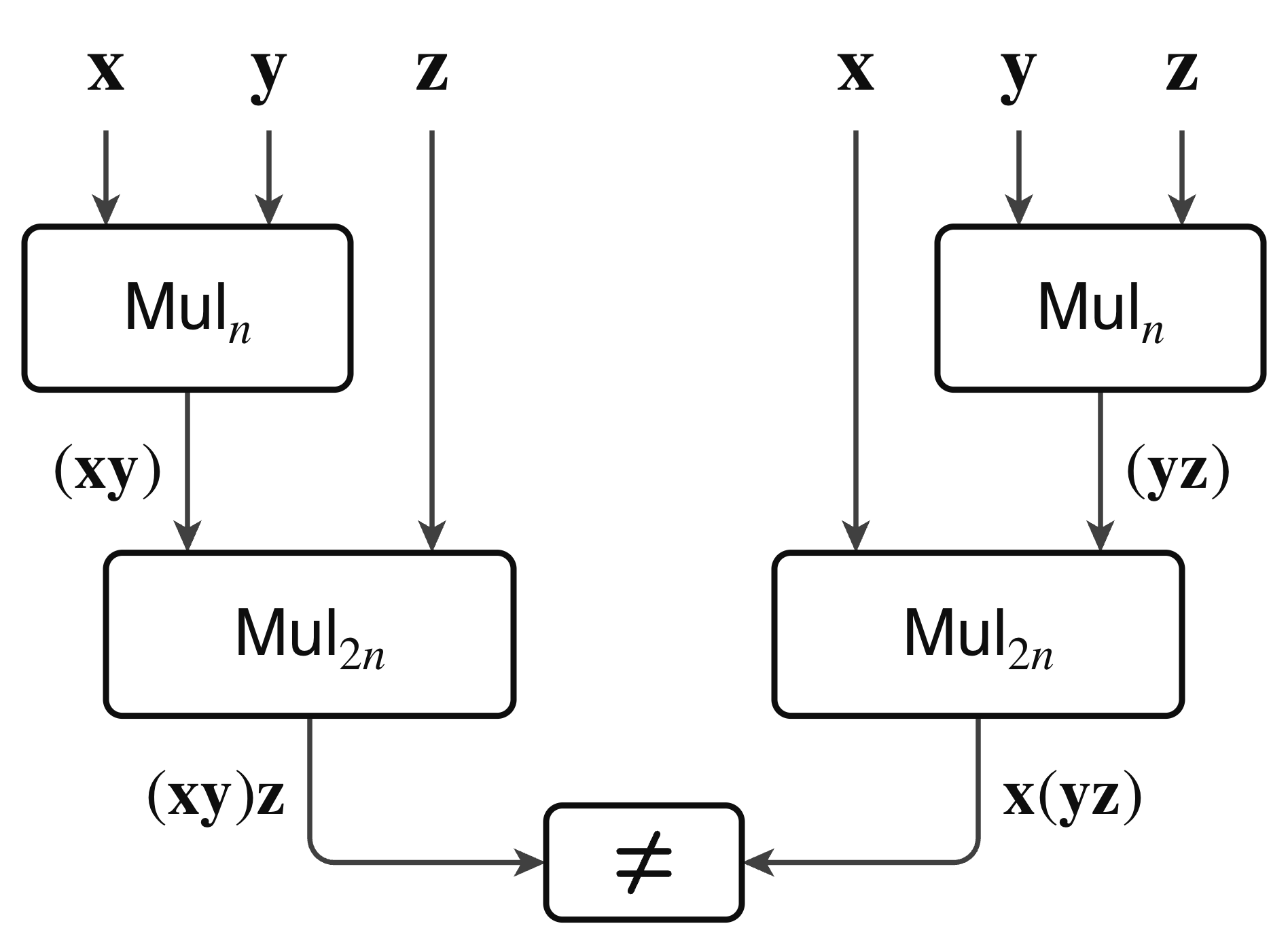}
\caption{The associativity instance $\Assoc_n$. The two inner multipliers $\Mul_n$ compute $\vv{x}\vv{y}$ and $\vv{y}\vv{z}$, the two outer
multipliers $\Mul_{2n}$ compute $(\vv{x}\vv{y})\vv{z}$ and $\vv{x}(\vv{y}\vv{z})$, and the miter variables $e_i$, with the long miter clause $e_0\vee\cdots\vee e_{4n-1}$, assert that $(\vv{x}\vv{y})\vv{z} \neq \vv{x}(\vv{y}\vv{z})$.}
\label{fig:assoc-miter}
\end{figure}

\section{Strip-local multipliers}
\label{sec:strip-local}

Our $\Assoc_n$ lower bound will apply so long as the multiplier encoding used is \emph{strip-local}, a property we will define shortly. Most multiplier encodings based on partial product summation satisfy this property, including standard array and Wallace-tree multipliers.

\begin{definition}[$g$-strip]
\label{def:g-strip}
For an integer $g\ge1$, the \emph{$g$-strip} with base $bg$ is the interval of
columns $[bg,(b+1)g)$, for $b=0,1,2,\ldots$. We call them simply \emph{strips} when $g$ is clear from context. A \emph{$g$-column} is
a column that is a multiple of $g$.
\end{definition}

\begin{definition}[$g$-sparse restriction]
\label{def:g-sparse}
For an integer $g\ge1$, a restriction of the input bit-vectors $\vv{x}$ and $\vv{y}$ is \emph{$g$-sparse} if every partial product $t_{i,j}$ not
forced to $0$ lies in a $g$-column (recall Definition~\ref{def:partial-products}) and each $g$-column contains strictly fewer than $2^{g-1}$ such partial products.
\end{definition}

\begin{definition}[Strip-local]
\label{def:strip-local}
An exact multiplier encoding (Definition~\ref{def:exact-encoding}) is \emph{strip-local} if its clauses have constant width and, under every
$g$-sparse restriction of the input bit-vectors, every wire is a Boolean function of the partial products belonging to a
single strip.
\end{definition}

\begin{remark}
	In the definition of $g$-sparse restriction, we set the threshold to $2^{g-1}$ instead of $2^g$ because Wallace-tree multipliers would not be strip-local with
	the latter threshold.
\end{remark}

\begin{observation}[Strips count their partial products]
\label{obs:strip-counts}
Fix an exact multiplier encoding, an integer $g\ge 1$, and a $g$-sparse restriction of the input bit-vectors. Then, under every assignment
to the unset inputs, the $g$ output bits belonging to the strip with base $bg$ compute, in binary, the number of partial products in the base
column $bg$ that evaluate to $1$.
\end{observation}

\begin{proposition}
\label{prop:array-strip-local}
The array multiplier encoding (Figure~\ref{fig:array-multiplier}) is strip-local.
\end{proposition}

\begin{proof}
Each variable of an array multiplier is determined by at most three other variables, so its clauses have width at most four.
 
Fix $g\ge1$ and a $g$-sparse restriction. In an array multiplier, the only wires crossing a column boundary are carries
(Figure~\ref{fig:array-multiplier}), so a wire in the strip with
base $bg$ can depend on lower strips only through the carries entering
column~$bg$.

It therefore suffices to show that every carry feeding into a $g$-strip is
the constant $0$. This is trivial for the strip at base $0$. Assume the inductive hypothesis that every carry feeding into the $g$-strip
with base $bg$ is zero. Then, since a carry into column $(b+1)g$ has weight $2^{(b+1)g}$, it can only be $1$ if the base column $bg$ of the previous strip
contained at least $2^{g}$ nonzero partial products, contradicting the $g$-sparsity of the restriction.
\end{proof}

\begin{restatable}{proposition}{wallacestriplocal}
\label{prop:wallace-strip-local}
The standard Wallace-tree multiplier encoding (Definition~\ref{def:wallace}), where the final two rows are summed by a
carry-lookahead adder (Definition~\ref{def:cla}), is strip-local.
\end{restatable}

The proof is nontrivial due to the final carry-lookahead adder. Appendix~\ref{app:architectures} gives the proof.

\section{Perfect matching and its star-local extension}
\label{sec:hard-source}

The source of hardness for our resolution lower bound is the sparse perfect-matching principle studied in~\cite{iss,ioss}. In general, a perfect-matching principle asserts that a given graph has no perfect matching. Its CNF encoding negates this: it has one variable per edge and requires that every vertex is incident to exactly one selected edge. When the underlying graph is bipartite and has a different number of vertices on the two sides, this CNF is unsatisfiable.

The best-known special case is the pigeonhole principle, which is the perfect-matching principle of the complete bipartite graph $K_{m+1,m}$. Haken showed that the pigeonhole principle is hard for resolution~\cite{haken}. However, this is an inconvenient source of hardness for our reduction due to the high degree of its vertices. To absorb the
auxiliary variables of the multiplier encodings, our reduction extends the formula with
a new variable for every Boolean function of the edges at a vertex. For $K_{m+1,m}$, this introduces a doubly exponential number of new variables, so the Ben-Sasson--Wigderson size--width tradeoff (Theorem~\ref{thm:bw}) does not give a useful size lower bound. On the other hand, for bounded-degree graphs, this extension only adds a constant number of variables and clauses per vertex, so the width lower bound of Itsykson, Slabodkin, and Sokolov~\cite{iss} implies an exponential size lower bound for the extension.

\begin{definition}[Perfect-matching CNF]
\label{def:pm}
For a bipartite graph $G=(U,V,E)$, introduce a variable $X_e$ for each edge $e\in E$. We interpret $X_e=1$ to mean that the edge $e$ is selected. The \emph{exact-one perfect-matching CNF} is
\[
\PM(G)
=
\bigwedge_{w\in U\cup V}
\operatorname{ExactOne}(X_{E(w)}),
\]
where $\operatorname{ExactOne}(X_{E(w)})$ consists of the clauses
\[
\bigvee_{e\in E(w)}X_e
\qquad\text{and}\qquad
\neg X_e\vee\neg X_f
\quad (e,f\in E(w),\ e\ne f).
\]
\end{definition}

For the remainder of the paper, fix an integer $m$, a constant $\Delta$ independent of $m$, and a bipartite graph
\[
H=(U,V,E),\qquad |U|=m+1,\qquad |V|=m,
\]
of maximum degree at most $\Delta$ and with no isolated vertices. We use a robust form of the hardness result of Itsykson, Slabodkin, and Sokolov~\cite{iss} for $\PM(H)$, based on star-local extension variables in the style of the functional encoding of~\cite{abrw}.

\begin{definition}[Star-local extension]
\label{def:star-local-ext}
For each vertex $w$ and Boolean function $f:\{0,1\}^{E(w)}\to\{0,1\}$, introduce an extension variable $Z_{w,f}$, defined by
clauses $\operatorname{Def}(Z_{w,f})$ over variables $\{Z_{w,f}\} \cup \{X_e:e\in E(w)\}$ expressing that $Z_{w,f}=f\bigl((X_e)_{e\in
E(w)}\bigr)$. The \emph{star-local extension} of the exact-one perfect-matching principle is
\[
\SPM(H)
=
\bigwedge_{w\in U\cup V}
\Bigl(
\operatorname{ExactOne}\bigl(X_{E(w)}\bigr)
\;\wedge\;
\bigwedge_{f:\{0,1\}^{E(w)}\to\{0,1\}}\operatorname{Def}(Z_{w,f})
\Bigr).
\]
\end{definition}

\begin{figure}[ht]
\centering
\includegraphics[width=0.62\textwidth]{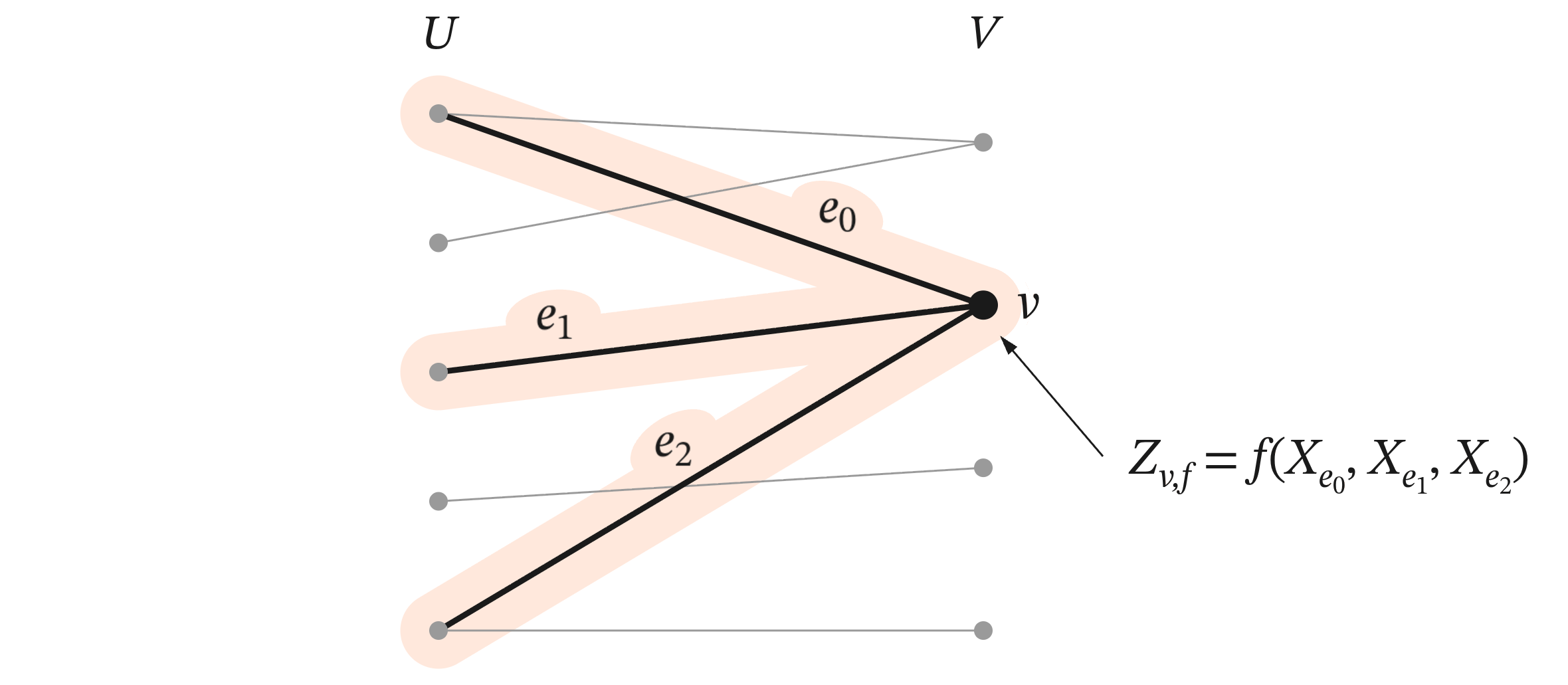}
\caption{A star-local extension variable. The vertex $v$ has star $E(v)=\{e_0,e_1,e_2\}$ (shaded), and for each Boolean function $f$ on this
star we introduce one variable $Z_{v,f}$ defined to equal $f(X_{e_0},X_{e_1},X_{e_2})$.}
\label{fig:star-local-ext}
\end{figure}

Figure~\ref{fig:star-local-ext} illustrates the extension at a single vertex. Every clause of $\SPM(H)$ has width at most $\Delta+1$. Because $\Delta$ is constant, each star has only a constant number of Boolean functions, so
$\SPM(H)$ still has $O(m)$ variables and $O(m)$ clauses.

\begin{definition}[Boundary expander]
\label{def:boundary-expander}
Let $G=(U,V,E)$ be a bipartite graph. For $S\subseteq U$, the \emph{boundary} $\delta(S)\subseteq V$ is the set of vertices of $V$ with exactly one neighbor in $S$. The graph $G$
is an \emph{$(r,\Gamma)$-boundary expander} if
\[
|\delta(S)|\ge \Gamma\cdot|S|
\qquad\text{for every }S\subseteq U\text{ with }|S|\le r.
\]
\end{definition}

\begin{theorem}[\cite{iss,ioss}]
\label{thm:iss-graphs}
There are absolute constants $\Delta\in\N$, $\Gamma\ge1$, and $\varepsilon>0$ such that for every sufficiently large $m$ there is a bipartite graph
$H=(U,V,E)$ with $|U|=m+1$ and $|V|=m$, of maximum degree at most $\Delta$, that is an $(\varepsilon m,\Gamma)$-boundary expander. Furthermore, $H$ has a matching
covering $V$.
\end{theorem}

The following lemma is straightforward, though to our knowledge it has not explicitly appeared in the literature. A close analogue appears as Theorem~13(ii) of Atserias and Bonet~\cite{ab04}, but they only considered extension variables representing conjunctions, rather than arbitrary Boolean functions. The proof itself is essentially the same as the proof of Lemma~4.3 given by Ben-Sasson and Nordstr\"om in~\cite{bn11}, although the statement of the lemma below is different from what they conclude.

\begin{lemma}[Width transfer]
\label{lem:width-transfer}
Let $\Phi$ be an unsatisfiable CNF and $F$ a set of Boolean functions. Let
\[
\Phi^{\mathrm{ext}}=\Phi\wedge\bigwedge_{f \in F}\operatorname{Def}(Z_f)
\]
be an extension of $\Phi$ with extension variables $Z_f$, where $\operatorname{Def}(Z_f)$ is a CNF expressing the constraint $Z_f=f(\vv{x}_f)$. If each tuple $\vv{x}_f$ consists of at most $\Delta$ variables of $\Phi$, then every resolution refutation of $\Phi^{\mathrm{ext}}$ has width at least $\frac{w(\Phi\vdash\bot)}{2\Delta}$.
\end{lemma}

\begin{proof}
Let $\pi^{\mathrm{ext}}$ be a resolution refutation of $\Phi^{\mathrm{ext}}$ of width $w$. We obtain a refutation $\pi$ of the CNF $\Phi$ of width at most $2\Delta w$ as follows. In each clause of $\pi^{\mathrm{ext}}$, replace each extension variable $Z_f$ by the Boolean function $f(\vv{x}_f)$. This turns each clause into a disjunction of at most $w$ Boolean functions of at most $\Delta$ variables of $\Phi$. Then replace each of these disjunctions by an equivalent CNF over its at most $\Delta w$ variables.

Since resolution is sound, in each step of $\pi^{\mathrm{ext}}$, the CNF of its resolvent is implied by the CNFs of its two premises. We obtain the refutation $\pi$ by using implicational completeness to fill in the intermediate resolution steps between these CNFs. As the two premise CNFs, together, have at most $2\Delta w$ variables, these intermediate steps have width at most $2\Delta w$.
\end{proof}

\begin{theorem}[The star-local extension remains hard]
\label{thm:star-local-hard}
Let $H$ be a bipartite expander given by Theorem~\ref{thm:iss-graphs}. Then $\SPM(H)$ requires resolution refutations of size $2^{\Omega(m)}$.
\end{theorem}

\begin{proof}
Theorem~1.1 of the full version of~\cite{iss} shows a resolution width lower bound for $\PM(H)$ of $\Omega(m)$. The maximum degree $\Delta$ of the graph $H$ is constant, so Lemma~\ref{lem:width-transfer} transfers this to a width lower bound for $\SPM(H)$ of $\Omega(m)$.
Since $\SPM(H)$ has $O(m)$ variables and constant width, the size--width tradeoff (Theorem~\ref{thm:bw}) gives a size lower bound of $2^{\Omega(m)}$.
\end{proof}

Our reduction will also need the following observation, which follows immediately from the implicational completeness of resolution.

\begin{lemma}[Local consequences have constant-size derivations]
	\label{lem:star-local-derivations}
	For each vertex $w \in U \cup V$, let $F_w$ denote the conjunction of the clauses of $\SPM(H)$ associated with $w$,
	\[
	F_w
	=
	\operatorname{ExactOne}\bigl(X_{E(w)}\bigr)
	\;\wedge\;
	\bigwedge_{f:\{0,1\}^{E(w)}\to\{0,1\}}\operatorname{Def}(Z_{w,f}),
	\]
	so that $\SPM(H)=\bigwedge_{w\in U\cup V}F_w$.

	Let $W\subseteq U\cup V$. If a clause $C$ is implied by
	$\bigwedge_{w\in W}F_w$, then either $C$ is a tautology, or
	some subclause $C' \subseteq C$ has a resolution derivation from
	$\bigwedge_{w\in W}F_w$ whose size depends only on $|W|$ and on $\Delta$.
\end{lemma}

\section{Overview of the reduction}
\label{sec:outline}

Throughout the reduction, we assume that $\Assoc_n$ uses a strip-local multiplier encoding $\Mul_n$ (Definition~\ref{def:strip-local}).
At a high level, the reduction from $\SPM(H)$ to $\Assoc_n$ works as follows. We begin by constructing a restriction $\rho$ of $\Assoc_n$ such that the remaining clauses $\Assoc_n\!\upharpoonright_{\rho}$
form a circuit encoding of the perfect-matching constraints. Because this encoding contains circuit variables that are not present in the perfect-matching formula $\PM(H)$, the original
perfect-matching lower bound does not immediately transfer to this restricted formula $\Assoc_n\!\upharpoonright_{\rho}$. This is why we introduced the star-local extension $\SPM(H)$. Each
clause of the restricted formula $\Assoc_n\!\upharpoonright_{\rho}$ has a constant-size resolution derivation from $\SPM(H)$ (after substituting star-local extension variables in place of the clause's wire variables). Therefore a resolution refutation of $\Assoc_n$ yields a resolution refutation of $\SPM(H)$ with only a constant-factor increase
in size. The exponential lower bound for $\SPM(H)$ (Theorem~\ref{thm:star-local-hard}) thereby transfers to $\Assoc_n$.
Figure~\ref{fig:refutation-shape} shows the shape of the composed refutation.

\begin{figure}[ht]
\centering
\includegraphics[width=0.75\textwidth]{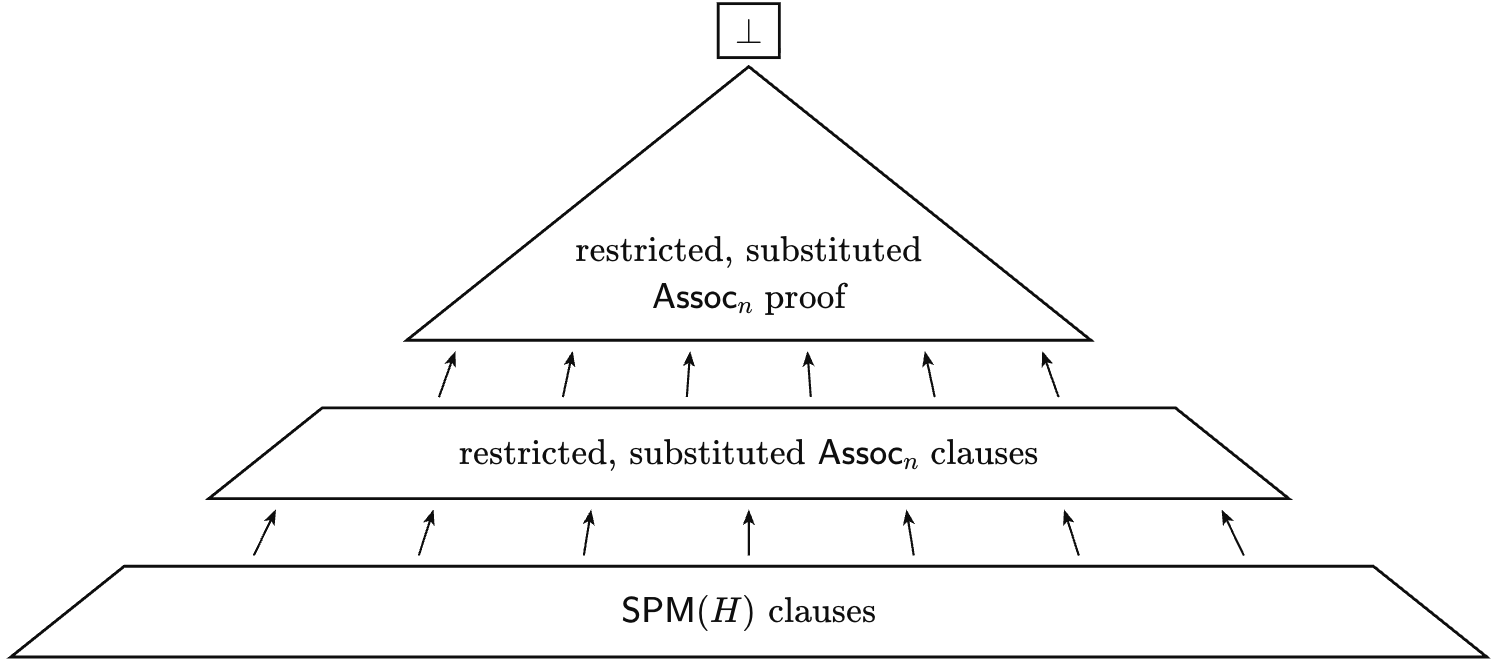}
\caption{The refutation of $\SPM(H)$, the star-local extension of perfect matching. This refutation uses $\SPM(H)$ to derive a subclause of each clause of the restricted, substituted
$\Assoc_n$ in a constant number of steps. The restricted, substituted refutation of $\Assoc_n$ then derives $\bot$ from these clauses.}
\label{fig:refutation-shape}
\end{figure}

\subsection{Reduction of star-local perfect matching to \texorpdfstring{$\Assoc_n$}{Assoc\_n}}
\label{sec:outline-reduction}

Suppose that $\pi$ is a resolution refutation of $\Assoc_n$ of size $s$. We convert $\pi$ into a refutation of $\SPM(H)$ of size $O(s)$ in
three steps.

\begin{itemize}
\item \textbf{(Restrict)} We construct a restriction $\rho$ that leaves one input bit of $\vv{y}$ unset for each edge $e\in E$ and
identifies it with the perfect-matching edge variable $X_e$. It sets the remaining input bits of $\vv{x}, \vv{y}, \vv{z}$ to constants $0$ or
$1$. The restriction $\rho$ also sets some of the output bits of the multipliers $\vv{x}\vv{y}$ and $\vv{y}\vv{z}$, and it sets one of the
miter variables $e_K=1$ at a fixed position $K$, satisfying the long miter clause $e_0\vee\cdots\vee e_{4n-1}$.

Under the restriction $\rho$, the two
inner multipliers of $\Assoc_n$ become encodings, with auxiliary variables, of the
perfect-matching constraints on the edge variables $X_e$ (Lemma~\ref{lem:inner-pm}). More precisely, under $\rho$ every wire computes a Boolean function of the edge variables, so an assignment $\alpha$ to the
edge variables $X_e$ propagates to an assignment $\bar\alpha$ of all remaining variables (Definition~\ref{def:propagation}). The propagated
assignment $\bar\alpha$ satisfies the restricted encoding $\vv{x}\vv{y}\!\upharpoonright_{\rho}$ if and only if the original edge variable assignment $\alpha$ satisfies
$\operatorname{ExactOne}(X_{E(v)})$ for every $v\in V$. Similarly, $\bar\alpha$ satisfies $\vv{y}\vv{z}\!\upharpoonright_{\rho}$ if and only if $\alpha$
satisfies $\operatorname{ExactOne}(X_{E(u)})$ for every $u\in U$ (Lemma~\ref{lem:inner-pm}).

Applying the restriction $\rho$ to the proof $\pi$ yields a refutation $\pi\!\upharpoonright_\rho$ of the restricted associativity formula
$\Assoc_n\!\upharpoonright_\rho$. 

The restricted formula $\Assoc_n\!\upharpoonright_\rho$ contains no multiplier input variables other than the edge variables $X_e$. Its
remaining variables are wires, which the next step eliminates.

\item \textbf{(Substitute)} This restriction $\rho$ will also ensure that every wire computes some function
$f(X_{E(w)})$ of the edge variables of a single star $E(w)$, $w \in U \cup V$ (Lemma~\ref{lem:locality}). We can therefore substitute each
wire throughout the refutation $\pi\!\upharpoonright_\rho$ by the corresponding extension variable $Z_{w,f}$.
After this substitution, every clause $C$ of the restricted associativity
formula $\Assoc_n\!\upharpoonright_\rho$ is implied by the clauses of $\SPM(H)$ associated
with a set of stars covering $C$'s variables (Lemma~\ref{lem:local-soundness}). Since each clause $C$ of $\Assoc_n\!\upharpoonright_\rho$ has constant width, its variables can be covered by a constant number of stars.

\item \textbf{(Compose)} Finally, we obtain a refutation of the star-local extension $\SPM(H)$ as follows. Using Lemma~\ref{lem:star-local-derivations}, we derive a subclause
of each restricted, substituted $\Assoc_n$ clause from the clauses of $\SPM(H)$ in only a constant number of steps. Then we
apply the restricted, substituted refutation of $\Assoc_n$ to refute the derived clauses.
\end{itemize}

\subsection{Outline of the restriction \texorpdfstring{$\rho$}{rho}}
\label{sec:outline-restriction}

We will define a $g$-sparse restriction $\rho$ in terms of three maps. We map each vertex $u \in U$ to an $\vv{x}$-bit-position $i(u)$, each edge $e \in E$ to a $\vv{y}$-bit-position $j(e)$, and each vertex $v
\in V$ to a $\vv{z}$-bit-position $k(v)$. For each vertex $u \in U$ and $v \in V$, the restriction $\rho$ sets $x_{i(u)} = 1$ and $z_{k(v)}
= 1$. For each edge $e \in E$, the restriction $\rho$ leaves $y_{j(e)}$ unrestricted and identifies it with the edge variable $X_e$. It
restricts all remaining bits of $\vv{x}$, $\vv{y}$, and $\vv{z}$ to $0$. Additionally, $\rho$ restricts some of the output bits of the
inner multipliers $\vv{x}\vv{y}$ and $\vv{y}\vv{z}$, together with one miter variable, as we will specify below.

\begin{figure}[htbp]
\centering
\includegraphics[width=0.75\textwidth]{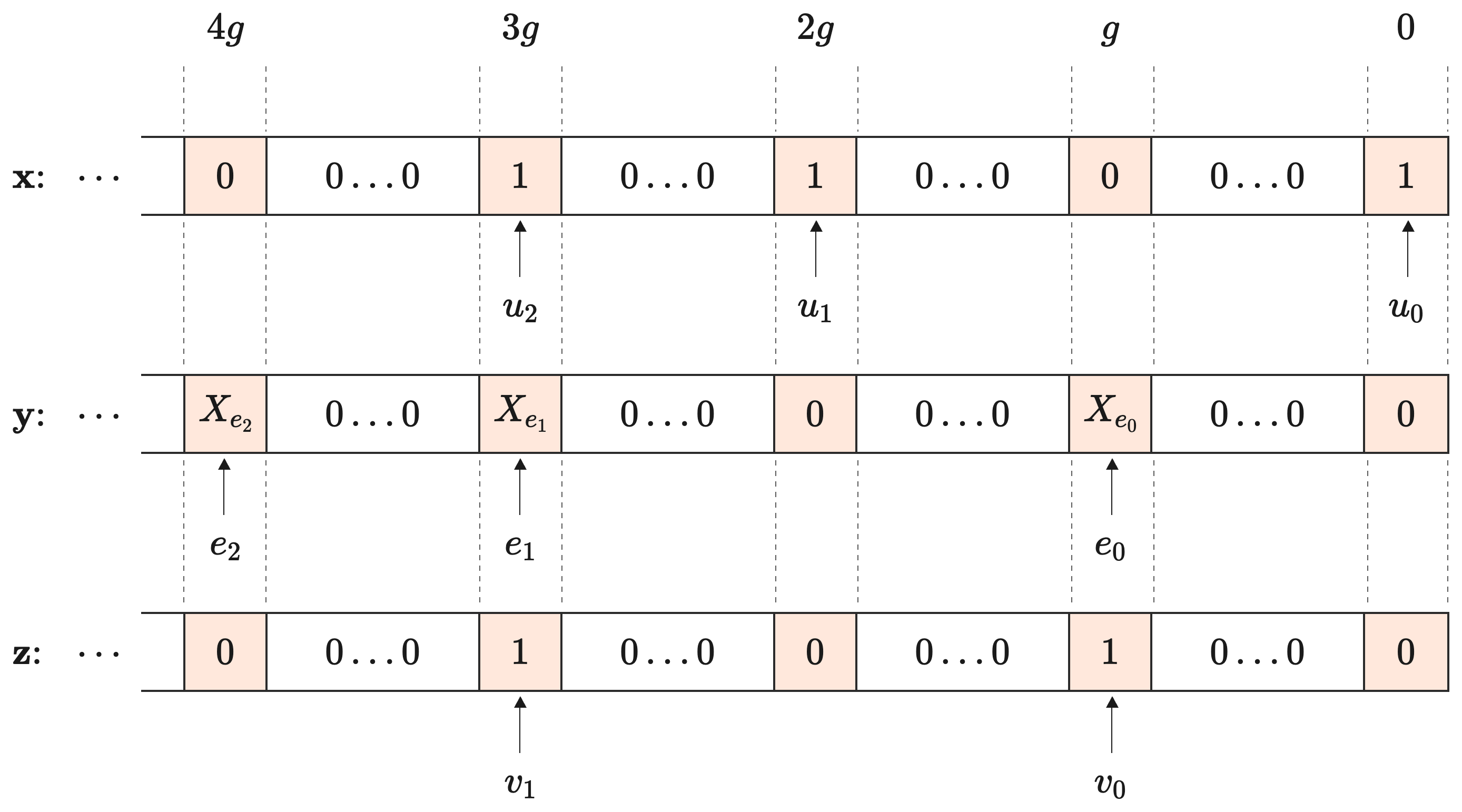}
\caption{The restricted input bit-vectors. Each vertex $u\in U$ maps to a $1$ in $\vv{x}$ at position $i(u)$, each vertex $v\in V$ maps to a $1$
in $\vv{z}$ at position $k(v)$, and each edge $e\in E$ maps to an unrestricted variable $y_{j(e)}$ in $\vv{y}$ that we identify with the
edge variable $X_e$. Each position is a multiple of $g$, and every other bit is fixed to $0$.}
\label{fig:bit-strings}
\end{figure}

Lemma~\ref{lem:indices} constructs the maps $i(u)$, $j(e)$, $k(v)$ with special collision properties.
In this outline we will use the fact that for an edge
$e=(u_e,v_e)$, this construction sets
\[
j(e)=K-i(u_e)-k(v_e),
\]
where $K$ is a fixed multiple of $g$. Under $\rho$, the four multipliers behave as follows.
\begin{itemize}
\item \textbf{($g$-columns)} For each vertex and edge, the values $i(u), j(e), k(v)$ are multiples of the strip width
$g=\lceil\log(m+2)\rceil+2$ (Lemma~\ref{lem:indices}). This ensures that the nonzero partial
products in each of the four multipliers occur only in $g$-columns
(Definition~\ref{def:g-strip}). Figure~\ref{fig:bit-strings} shows the restricted input vectors $\vv{x}$, $\vv{y}$, and $\vv{z}$. Moreover, every $g$-column will contain fewer than $2^{g-1}$ of these nonzero partial products, making $\rho$
a $g$-sparse restriction of each of the four multipliers (Lemma~\ref{lem:rho-g-sparse}). By strip-locality, every wire is a
Boolean function of the partial products appearing in a single $g$-strip.

\item \textbf{(Inner multipliers)} Each vertex $v \in V$ is associated with a unique $g$-column of the multiplier $\vv{x}\vv{y}$, whose index is
$\col_{\vv{xy}}(v) := K - k(v)$. This column contains the partial products that, under $\rho$, equal the variables $X_e$ for edges $e$ of the star $E(v)$. So the effect of the multiplier circuit $\vv{x}\vv{y}$ is to constrain the output of the $g$-strip based at $\col_{\vv{xy}}(v)$
to equal the sum $\sum_{e \in E(v)} X_e$ represented in binary.

For each $v\in V$, the restriction $\rho$ fixes the output of the
$g$-strip based at $\col_{\vv{xy}}(v)$ to the binary representation of
$1$. This enforces $\sum_{e\in E(v)}X_e=1$, which is equivalent to the constraint $\operatorname{ExactOne}(X_{E(v)})$. Figure~\ref{fig:inner-multipliers} illustrates the effect of the restriction $\rho$ in the multiplier $\vv{x}\vv{y}$.

Symmetrically, each $u \in U$ is associated with a unique $g$-column of the multiplier $\vv{y}\vv{z}$, whose index is
$\col_{\vv{yz}}(u) := K - i(u)$, and which contains the variables $X_e$ for edges $e$ of the star $E(u)$. The restriction $\rho$ fixes the
output of the $g$-strip based at $\col_{\vv{yz}}(u)$ to the binary representation of $1$, enforcing $\operatorname{ExactOne}(X_{E(u)})$.

\begin{figure}[htbp]
\centering
\includegraphics[width=0.9\textwidth]{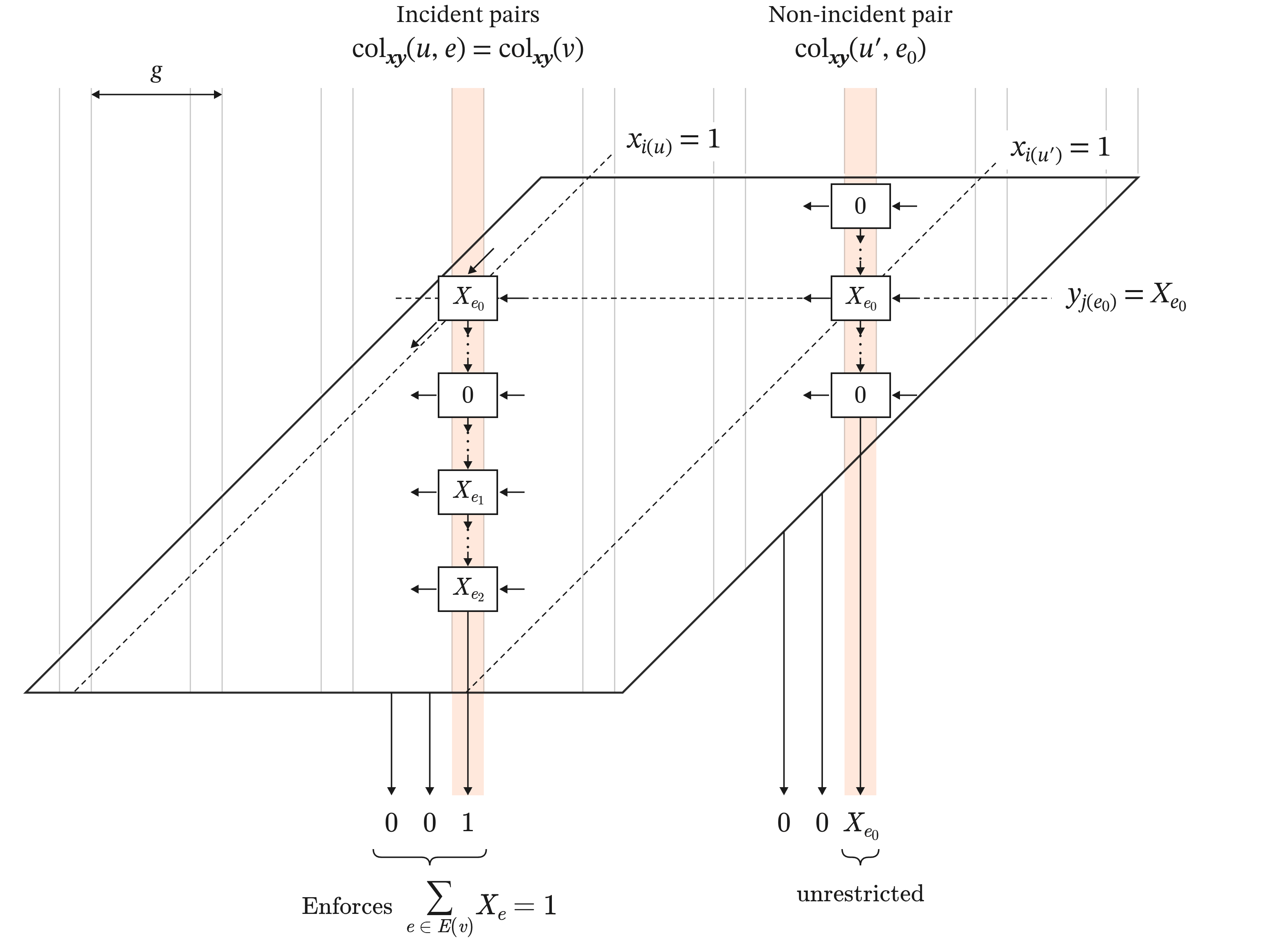}
\caption{An inner array multiplier $\vv{x}\vv{y}$ under the restriction $\rho$. Incident pairs $(u_e,e)$ with $e\in E(v)$ map to the common column $\col_{\vv{xy}}(u_e,e)=\col_{\vv{xy}}(v)$, so this column contains the
partial products $X_e$ for all edges $e\in E(v)$. The restriction $\rho$ fixes the output of this column's $g$-strip to the binary representation of $1$, enforcing
$\sum_{e\in E(v)}X_e=1$. Each non-incident pair $(u,e)$ maps to its own column $\col_{\vv{xy}}(u,e)$, so the output of this column is $X_e$.}
\label{fig:inner-multipliers}
\end{figure}

\item \textbf{(Outer multipliers)} Column $K$ of the outer multiplier $(\vv{x}\vv{y})\vv{z}$ contains exactly one $1$ for each $v \in V$, as shown in Figure~\ref{fig:outer-multiplier}: the restriction $\rho$ sets the output bit of $\vv{x}\vv{y}$ at index $\col_{\vv{xy}}(v)=K-k(v)$ to $1$, and multiplying this output by
$z_{k(v)}=1$ contributes a $1$ to column $K$. The remaining partial products in column $K$ are $0$ since $\rho$ restricts every other bit of $\vv{z}$ to $0$.

Therefore column $K$ of
$(\vv{x}\vv{y})\vv{z}$ contains exactly $|V|=m$ ones, so the output of the $g$-strip based at
$K$ is forced to the binary representation of $m$ (Observation~\ref{obs:strip-counts}).
Symmetrically, column $K$ of the outer multiplier $\vv{x}(\vv{y}\vv{z})$ contains a $1$ for each $u
\in U$, so the output of the $g$-strip based at $K$ is forced to the binary representation of $m+1$.
This means that the outer multipliers disagree on their $K$-th output bit, so $\rho$ sets the miter variable $e_K = 1$ to satisfy the long
miter clause $e_0 \vee \cdots \vee e_{4n-1}$. This completes the description of $\rho$.

\item \textbf{(Other $g$-columns)} In all four multipliers, each remaining $g$-column not covered above contains edge variables from at most one star (Figures~\ref{fig:inner-multipliers} and~\ref{fig:outer-star-local}). Therefore every wire in the corresponding $g$-strips is a Boolean function of a single star.
\end{itemize}

\begin{figure}[t]
\centering
\includegraphics[width=0.9\textwidth]{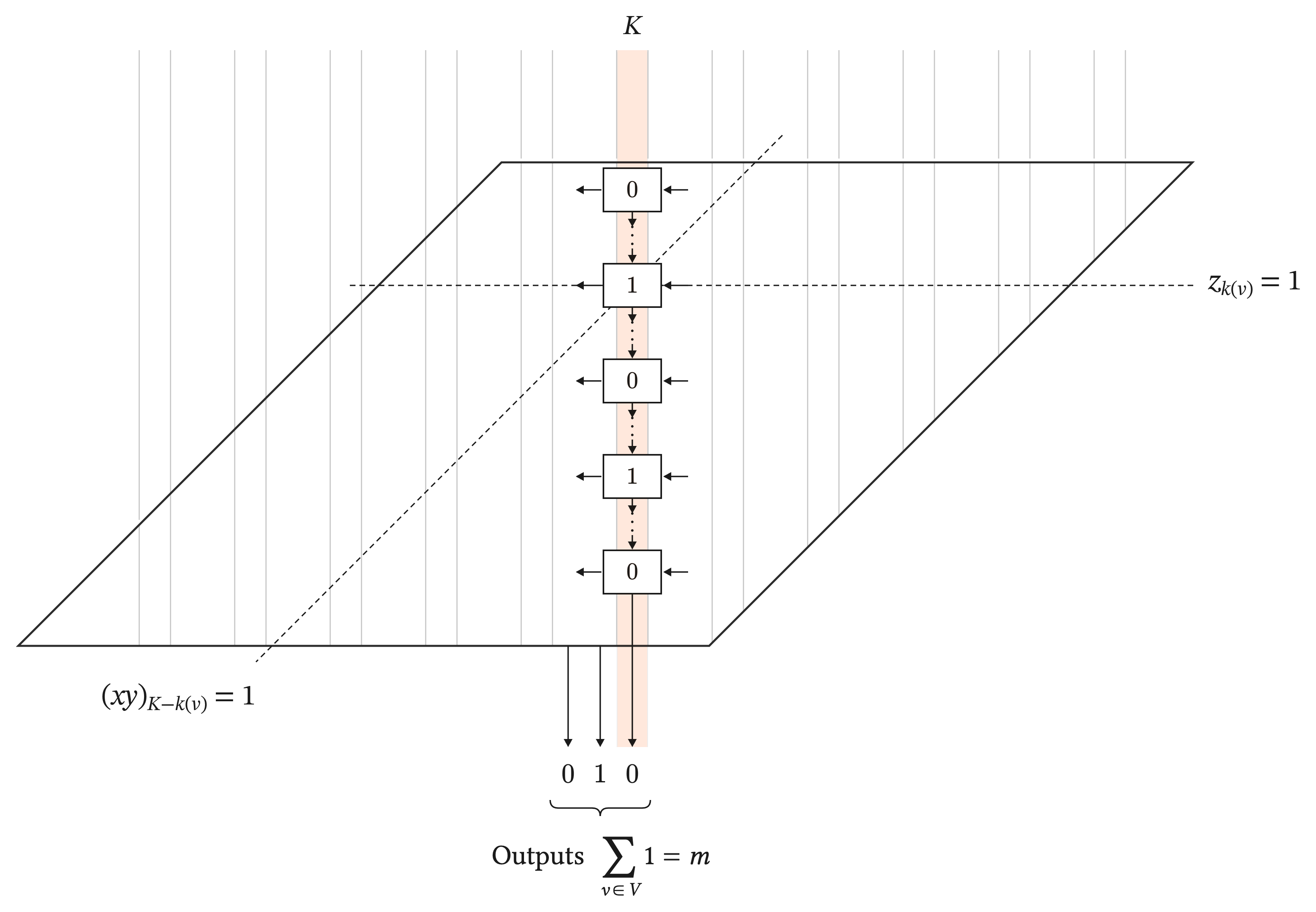}
\caption{An outer array multiplier $(\vv{x}\vv{y})\vv{z}$ under the restriction $\rho$. Column $K$ contains a constant $1$ for each $v\in V$, which comes from the output bit of $\vv{x}\vv{y}$ at $\col_{\vv{xy}}(v)$ multiplied by $z_{k(v)}=1$. Every other partial product in
the column is $0$, so the $g$-strip based at $K$ outputs the binary representation of $m$. In the other outer multiplier
$\vv{x}(\vv{y}\vv{z})$, the same $g$-strip outputs $m+1$.}
\label{fig:outer-multiplier}
\end{figure}

\section{Construction of the restriction \texorpdfstring{$\rho$}{rho}}
\label{sec:rho}

In this section we build the index maps embedding the bipartite graph $H=(U,V,E)$ into the input bit-vectors $\vv{x},\vv{y},\vv{z}$. We will then use these maps to define our $g$-sparse restriction $\rho$.

\subsection{Construction of the index maps \texorpdfstring{$i(u)$, $j(e)$, $k(v)$}{i(u), j(e), k(v)}}
\label{sec:embedding}

We build
three maps: $i:U\to\N$ for $\vv{x}$-bit-positions, $j:E\to\N$ for $\vv{y}$-bit-positions, and $k:V\to\N$ for $\vv{z}$-bit-positions. Recall
from Section~\ref{sec:outline-restriction} that the restriction $\rho$ will set $x_{i(u)}=1$ and $z_{k(v)}=1$, identify each $y_{j(e)}$ with
the edge variable $X_e$, and restrict all remaining bits of $\vv{x}$, $\vv{y}$, and $\vv{z}$ to $0$.

Thus $u\in U$, $e\in E$, and $v\in V$ index bits of $\vv{x}$,
$\vv{y}$, and $\vv{z}$ respectively. Pairs and triples of them point to
partial products in the multipliers: the pair $(u,e)$ points to the partial product
$x_{i(u)}y_{j(e)}$ of the inner multiplier $\vv{x}\vv{y}$, in column
$
i(u)+j(e),
$
the pair $(e,v)$ points to $y_{j(e)}z_{k(v)}$ of $\vv{y}\vv{z}$, in
column $j(e)+k(v)$, and the triple $(u,e,v)$
points to an outer partial product in each of $(\vv{x}\vv{y})\vv{z}$
and $\vv{x}(\vv{y}\vv{z})$, both in the same column
$
i(u)+j(e)+k(v).
$
We call $(u,e)$, $(e,v)$, and $(u,e,v)$ \emph{incident} when the vertices are endpoints of $e$.
Accordingly, we define maps from pairs and triples to columns.
\begin{definition}[Columns]
	\label{def:columns}
	For maps $i:U\to\N$, $j:E\to\N$, and $k:V\to\N$ and
	$u\in U$, $e\in E$, and $v\in V$, let
	\[
	\col_{\vv{xy}}(u,e):=i(u)+j(e),\qquad
	\col_{\vv{yz}}(e,v):=j(e)+k(v),
	\]
	\[
	\col_{(\vv{x}\vv{y})\vv{z}}(u,e,v):=i(u)+j(e)+k(v),\qquad
	\col_{\vv{x}(\vv{y}\vv{z})}(u,e,v):=i(u)+j(e)+k(v).
	\]
\end{definition}

We build the maps $i(u)$, $j(e)$, $k(v)$ from a Golomb ruler, which prevents unintended collisions when we map vertices and edges to $g$-strips.

\begin{lemma}[Directed Golomb ruler]
\label{lem:golomb}
For every $m\ge1$, set $M=32m^2$. There are $m+1$ integers $s_0,\ldots,s_m\in[0,M)$ called \emph{marks} such that
\[
s_a-s_b=s_c-s_d\ne0\quad\Longrightarrow\quad (a,b)=(c,d).
\]
In other words, every nonzero directed difference is produced by a unique ordered pair of marks.
\end{lemma}

\begin{proof}
Use the Erd\H{o}s--Tur\'an construction~\cite{erdos-turan}. Choose an odd prime $m+1<p<2(m+1)$ and set
\[
s_r=2pr+(r^2\bmod p),\qquad 0\le r<p.
\]
Since $s_r<2p^2<8(m+1)^2\le32m^2=M$, every mark lies in the required range. Take the first $m+1$ marks $s_0,\ldots,s_m$. Put $\tau_r:=r^2\bmod p\in\{0,\ldots,p-1\}$. If
$s_a-s_b=s_c-s_d\ne0$, then
\[
2p((a-b)-(c-d))=(\tau_c-\tau_d)-(\tau_a-\tau_b).
\]
The right side has absolute value $<2p$, while the left side is a multiple of $2p$, so $a-b=c-d=:h$. Reducing the original equality modulo
$p$ gives $h(a+b)\equiv h(c+d)$. Since the difference is nonzero, $h\ne0$ and $|h|<p$, so cancel $h$ to get $a+b\equiv c+d\pmod p$. Adding
and subtracting this congruence with $a-b=c-d$ gives $2a\equiv2c$ and $2b\equiv2d\pmod p$. Because $p$ is odd, $2$ is invertible, hence
$a\equiv c$ and $b\equiv d\pmod p$. As $a,b,c,d\in\{0,\ldots,p-1\}$, we have $a=c$ and $b=d$.
\end{proof}

\begin{lemma}[Construction of index maps]
	\label{lem:indices}
Let $H=(U,V,E)$ be a bipartite graph with $|U|=m+1$ and $|V|=m$, and set
\[
g=\lceil\log(m+2)\rceil+2.
\]
Then there is an integer $n_0=O(m^4\log m)$ as well as maps
\[
i:U\to\N,\qquad j:E\to\N,\qquad k:V\to\N
\]
with the following properties.
\begin{enumerate}[label=(\roman*)]
\item \textbf{Mapping pattern.} The maps $i$, $j$, and $k$ map injectively to multiples of $g$. Moreover, $i(u)+j(e)+k(v)<n_0$ for all $u\in U$, $e\in E$, and $v\in V$.
\item \textbf{Inner collision pattern.}
\begin{enumerate}[label=(\alph*)]
\item For each vertex $v\in V$ the
incident pairs $(u_e,e)$ with edge $e\in E(v)$ map to a common column $\col_{\vv{xy}}(u_e,e)$, denoted
$\col_{\vv{xy}}(v)$. Distinct vertices $v\in V$ map to distinct columns $\col_{\vv{xy}}(v)$. Each non-incident pair $(u,e)$ occupies its own column $\col_{\vv{xy}}(u,e)$.
\item For each vertex $u\in U$ the
incident pairs $(e,v_e)$ with edge $e\in E(u)$ map to a common column $\col_{\vv{yz}}(e,v_e)$, denoted
$\col_{\vv{yz}}(u)$. Distinct vertices in $U$ map to distinct columns $\col_{\vv{yz}}(u)$. Each non-incident pair $(e,v)$ occupies its own column $\col_{\vv{yz}}(e,v)$.
\end{enumerate}
\item \textbf{Outer collision pattern.} Define the \emph{difference pair} of a triple
$(u,e,v)\in U\times E\times V$ as
\[
\bigl(\,i(u)-i(u_e),\;k(v)-k(v_e)\,\bigr).
\]
The column $\col_{(\vv{x}\vv{y})\vv{z}}(u,e,v)$ depends only on the difference pair of the
triple, and distinct difference pairs map to distinct columns. Furthermore:
\begin{enumerate}[label=(\alph*)]
\item The triples with difference pair $(0,0)$ are exactly the incident triples
$(u_e,e,v_e)$, $e\in E$. We denote their common column by $K$.
\item Each nonzero difference determines its ordered pair of vertices: if
$i(u)-i(u_e)=i(u')-i(u_{e'})\ne0$ then $(u,u_e)=(u',u_{e'})$, and if
$k(v)-k(v_e)=k(v')-k(v_{e'})\ne0$ then $(v,v_e)=(v',v_{e'})$.
\end{enumerate}
\end{enumerate}
\end{lemma}

Before giving the proof, we describe what properties (ii) and (iii) will mean for the $g$-strips of the four multipliers under $\rho$. Property
(ii) ensures that in the inner multipliers every $g$-column holds either the edge variables of a single star $E(u)$ or $E(v)$, or at most one edge variable $X_e$ (Figure~\ref{fig:inner-multipliers}). Since $\rho$ will fix the output of each single-star strip to the binary representation of $1$
(Section~\ref{sec:outline-restriction}), the inner multipliers then enforce exactly the constraints
$\operatorname{ExactOne}(X_{E(v)})$ and $\operatorname{ExactOne}(X_{E(u)})$, for $v\in V$ and $u\in U$.

Property (iii) assigns each difference pair to a unique $g$-column at the same position in both outer multipliers. This serves two purposes.
First, by (iii)(a) the partial products lying in the $g$-column for the difference pair $(0,0)$ are exactly those corresponding to incident triples
$(u_e,e,v_e)$. As
described in Section~\ref{sec:outline-restriction}, the restriction $\rho$ forces these partial products to constants, which then forces the strip at $K$ to output $m$ in
$(\vv{x}\vv{y})\vv{z}$ and $m+1$ in $\vv{x}(\vv{y}\vv{z})$ (Figure~\ref{fig:outer-multiplier}).
Second, (iii)(b) ensures that each remaining outer multiplier $g$-column contains edge variables from at most one star (Figure~\ref{fig:outer-star-local}). For example, consider the case of a $g$-column corresponding to the difference pair
$(d, 0)$ with $d \ne 0$. The difference $d$ is produced by a unique
ordered pair $(u, u_e)$ with $d = i(u) - i(u_e)$, so only edge variables $X_e$ for edges $e \in E(u_e)$ can appear
in this $g$-column. The same reasoning applies to every $g$-column corresponding to a non-$(0,0)$ difference pair.

\begin{figure}[ht]
\centering
\includegraphics[width=0.95\textwidth]{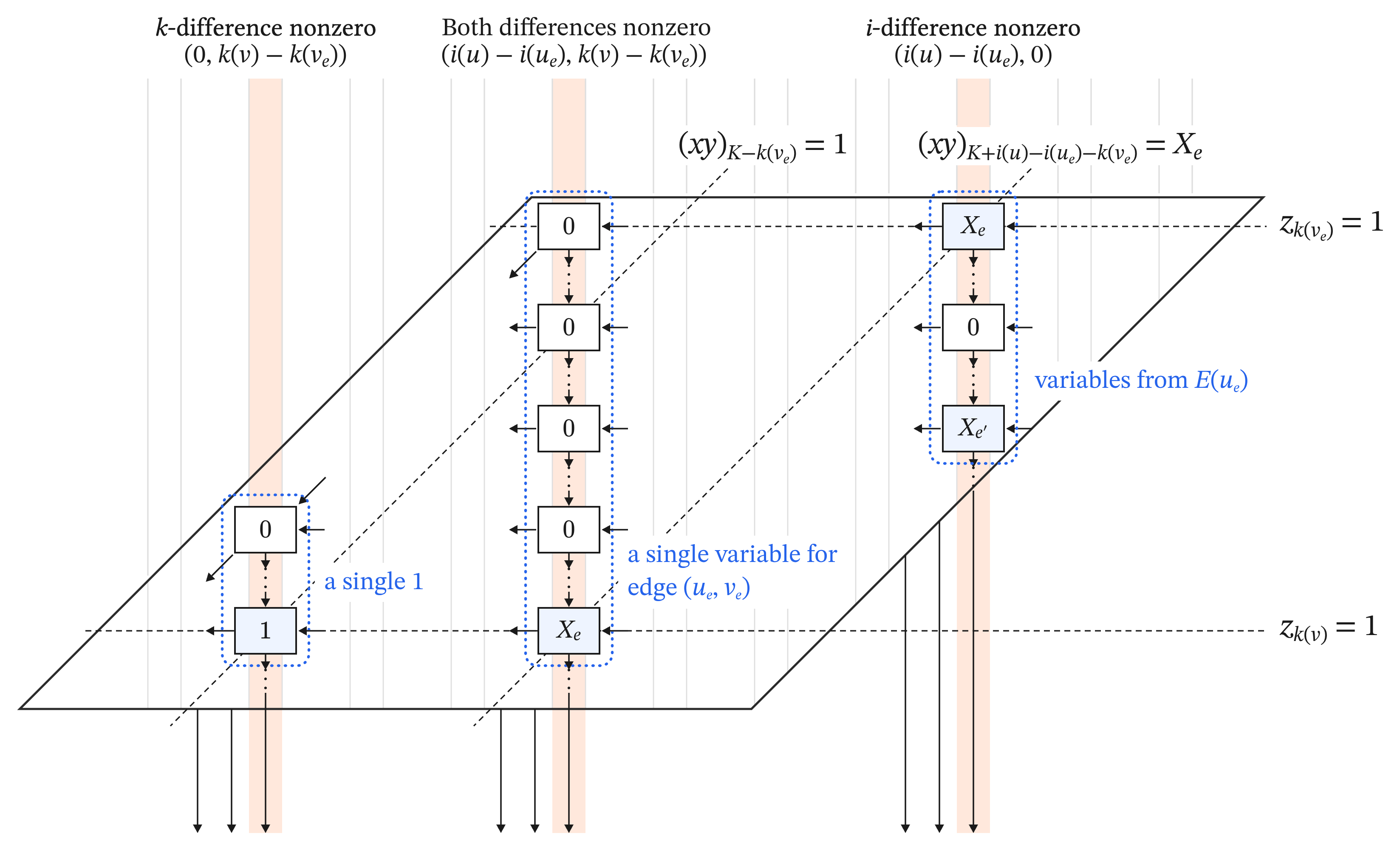}
\caption{Aside from column $K$, there are three kinds of $g$-columns of the outer multiplier $(\vv{x}\vv{y})\vv{z}$ under the restriction $\rho$. In each case, the wires inside the corresponding $g$-strip are functions of the Boolean variables from a single star.}
\label{fig:outer-star-local}
\end{figure}

\begin{proof}
With $M=32m^2$, let $s_0,\ldots,s_m\in[0,M)$ be the marks given by Lemma~\ref{lem:golomb}. Choose injections of $U$ and $V$ separately into these marks,
and write $s_u$ and $s_v$ for the marks assigned to $u\in U$ and $v\in V$. To obtain property (iii), we want the outer
column
\[
\col_{(\vv{x}\vv{y})\vv{z}}(u,e,v)=i(u)+j(e)+k(v)
\]
to be uniquely determined by the two directed differences $s_u-s_{u_e}$ and $s_v-s_{v_e}$. Since both differences have magnitude $<M$, set
\[
B=2M+1.
\]
Then the two-digit base-$B$ value $a+Bb$, for $|a|,|b|<M$, uniquely determines its digits. We refer to this as \emph{base-$B$ uniqueness}.
We therefore scale the $U$-marks by $g$ and the $V$-marks by $gB$, and build $j(e)$ so that $\col_{(\vv{x}\vv{y})\vv{z}}$ is the same column
for all incident triples $(u_e,e,v_e)$. To this end, set
\begin{equation}
	\label{eq:coords}
	i(u)=gs_u,\qquad k(v)=gBs_v,\qquad j(e)=K-i(u_e)-k(v_e)
\end{equation}
where we choose the constant $K=2g(B+1)M$ to make $j(e)$ non-negative. Set $n_0:=4K$. With these
parameters, $n_0=O(m^4\log m)$.

\emph{Proof of (i).} It is clear from the construction that each $i(u), j(e), k(v)$ is a multiple of $g$. The maps $i$ and $k$ are injective because the marks $s_u$ and $s_v$ are assigned injectively. For the map $j$, since
\[
j(e)=K-g(s_{u_e}+Bs_{v_e}),
\]
base-$B$ uniqueness shows that $j(e)=j(e')$ implies $u_e=u_{e'}$ and $v_e=v_{e'}$, hence $e=e'$.
Finally, $i(u)+j(e)+k(v)=K+g(s_u-s_{u_e})+gB(s_v-s_{v_e})$ exceeds $K$ by less than $g(B+1)M=K/2$, so it is less than $3K/2<n_0$.

\emph{Proof of (ii).} For the maps $i(u)$ and $j(e)$ as defined in \eqref{eq:coords}, we have
\begin{equation}
\label{eq:left-col}
 \col_{\vv{xy}}(u,e)= i(u) + j(e) = K-k(v_e)+g(s_u-s_{u_e}).
\end{equation}
Thus, for each $v\in V$, the incident pairs $(u_e,e)$ with $e\in E(v)$ share the column $\col_{\vv{xy}}(v):=K-k(v)$, and these columns
are distinct for distinct $v$ because $k$ is injective. Each non-incident pair $(u,e)$ occupies its own
column because of base-$B$ uniqueness and the fact that the difference $s_u-s_{u_e}$ is attained by a unique pair $(u, u_e)$
(Lemma~\ref{lem:golomb}). This proves (ii)(a).

Symmetrically,
\begin{equation}
\label{eq:right-col}
 \col_{\vv{yz}}(e,v)=K-i(u_e)+gB(s_v-s_{v_e}),
\end{equation}
and the same argument proves (ii)(b).

\emph{Proof of (iii).} For the maps $i(u), j(e), k(v)$ defined in \eqref{eq:coords}, we have
\begin{align*}
\col_{(\vv{x}\vv{y})\vv{z}}(u,e,v)&=K+\bigl(i(u)-i(u_e)\bigr)+\bigl(k(v)-k(v_e)\bigr)\\
&=K+g(s_u-s_{u_e})+gB(s_v-s_{v_e}).
\end{align*}
The first line shows that the column depends only on the difference pair of the triple.
The second line shows, by base-$B$ uniqueness, that two columns $\col_{(\vv{x}\vv{y})\vv{z}}(u,e,v)$ and $\col_{(\vv{x}\vv{y})\vv{z}}(u',e',v')$ are equal if and only
if they agree on both differences: $s_u - s_{u_e} = s_{u'} - s_{u_{e'}}$ and $s_v-s_{v_e} = s_{v'}-s_{v_{e'}}$. Hence, distinct difference pairs map to distinct columns.

For (a), a triple has difference pair $(0,0)$ if and only if
$u=u_e$ and $v=v_e$, that is, the triple is incident.

For (b), suppose that $i(u)-i(u_e)=i(u')-i(u_{e'})\ne0$. Then
$s_u-s_{u_e}=s_{u'}-s_{u_{e'}}\ne0$, so by Lemma~\ref{lem:golomb} we must have $u=u'$ and $u_e =
u_{e'}$. The case of a nonzero $k$-difference is symmetric.
\end{proof}

The rest of the argument uses only the three properties of Lemma~\ref{lem:indices}.

\subsection{The restriction \texorpdfstring{$\rho$}{rho}}
\label{sec:restricted}

Using the maps $i(u), j(e), k(v)$ provided by Lemma~\ref{lem:indices}, we now give the full description of the restriction $\rho$. We then prove
that it is $g$-sparse, and that it forces the outputs of the outer multipliers to disagree at column $K$.

\begin{definition}[The restriction $\rho$]
\label{def:rho}
Let the maps $i(u)$, $j(e)$, $k(v)$ be as in Lemma~\ref{lem:indices}, let $g=\lceil\log(m+2)\rceil+2$ be the strip width, and let $K$ be
the column given by Lemma~\ref{lem:indices}(iii).
\begin{enumerate}[label=(\alph*)]
\item \textbf{(Inputs)} For each vertex $u\in U$ and $v\in V$, the restriction $\rho$ sets
$x_{i(u)}=1$ and $z_{k(v)}=1$. For each edge $e\in E$, $\rho$ leaves the bit $y_{j(e)}$
unset and identifies it with the edge variable $X_e$; it restricts all
remaining bits of $\vv{x}$, $\vv{y}$, and $\vv{z}$ to $0$.
\item \textbf{(Inner multiplier outputs)} Recall that Lemma~\ref{lem:indices}(ii) provides well-defined and injective column maps
$\col_{\vv{xy}}(v)$ for $v\in V$ and $\col_{\vv{yz}}(u)$ for $u\in U$.
For each $v\in V$, the restriction $\rho$ fixes the output of the $g$-strip of $\vv{x}\vv{y}$ with base $\col_{\vv{xy}}(v)$ to the binary representation of
$1$ ($\rho$ sets the output bit in column $\col_{\vv{xy}}(v)$ to $1$ and the other $g-1$ output bits to $0$), and for each $u\in U$, it
fixes the output of the $g$-strip of $\vv{y}\vv{z}$ with base $\col_{\vv{yz}}(u)$ to the binary representation of $1$. The restriction $\rho$ leaves the output bit in column $\col_{\vv{xy}}(u,e)$ unset for each non-incident pair $(u,e)$, and sets every other output
bit of $\vv{x}\vv{y}$ to $0$. Symmetrically, $\rho$ leaves the output bit in column
$\col_{\vv{yz}}(e,v)$ unset for each non-incident pair $(e,v)$, and sets every other output bit of $\vv{y}\vv{z}$ to $0$.
\item \textbf{(Miter)} The restriction $\rho$ sets the miter variable $e_K=1$.
\end{enumerate}
\end{definition}

\begin{lemma}[$\rho$ is a $g$-sparse restriction]
\label{lem:rho-g-sparse}
The following hold.
\begin{enumerate}[label=(\roman*)]
\item Every partial product of $\vv{x}\vv{y}$ not forced to $0$ lies in a $g$-column. Additionally, for each $v\in V$, the $g$-column $\col_{\vv{xy}}(v)$
contains exactly the partial products
\[
x_{i(u_e)}y_{j(e)}=X_e,\qquad e\in E(v),
\]
and every other $g$-column contains at most one partial product not forced to $0$; symmetrically for $\vv{y}\vv{z}$ with the $g$-columns
$\col_{\vv{yz}}(u)$, $u\in U$.
\item For each of the four multipliers, $\rho$ is a $g$-sparse restriction of its input bit-vectors.
\end{enumerate}
\end{lemma}

\begin{proof}
(i) Consider the inner multiplier $\vv{x}\vv{y}$. A partial product $x_iy_j$ is not forced to $0$ by the restriction $\rho$ if and only if $x_i$ is
set to $1$ and $y_j$ is not set to $0$, that is, $i=i(u)$ and $j=j(e)$ for some $u\in U$ and $e\in E$. The partial product then computes
$x_{i(u)}y_{j(e)}=X_e$ and lies in column $\col_{\vv{xy}}(u,e)$, which is a $g$-column by Lemma~\ref{lem:indices}(i). By Lemma~\ref{lem:indices}(ii)(a), the pairs $(u,e)$ with $\col_{\vv{xy}}(u,e)=\col_{\vv{xy}}(v)$ are exactly the incident pairs
$(u_e,e)$ with $e\in E(v)$, while every non-incident pair occupies a $g$-column of its own. Hence the $g$-column $\col_{\vv{xy}}(v)$ contains exactly the partial products $X_e$, $e\in E(v)$, and every other $g$-column contains at
most one partial product $X_e$. The argument for $\vv{y}\vv{z}$ is symmetric.

(ii) Since $2^{g-1}\ge2(m+2)>m+1$, it suffices to show that the partial products of each multiplier
not forced to $0$ lie in $g$-columns, with at most $m+1$ such partial products per $g$-column. For the
inner multipliers this follows from (i), since every star of $H$ has at most $m+1$ edges.

Now consider the outer multiplier $(\vv{x}\vv{y})\vv{z}$. The restriction $\rho$ sets each of the bits in $\vv{z}$ to $0$ except for the bits $z_{k(v)}$, and sets each of the bits in $\vv{x}\vv{y}$ to $0$ except for those at the columns $\col_{\vv{xy}}(u,e)$. So every partial
product of $(\vv{x}\vv{y})\vv{z}$ not forced to $0$ lies in a column of the form $\col_{\vv{xy}}(u,e)+k(v)=\col_{(\vv{x}\vv{y})\vv{z}}(u,e,v)$. This is a $g$-column by
Lemma~\ref{lem:indices}(i). Within a given column of $(\vv{x}\vv{y})\vv{z}$, distinct partial products have distinct factors $z_{k(v)}$ and hence correspond to distinct $v\in V$. So each $g$-column contains at most $|V|=m$ partial products not forced to $0$. The same holds for the $g$-columns of
$\vv{x}(\vv{y}\vv{z})$ with the $m+1$ vertices $u \in U$.
\end{proof}

\begin{lemma}[$\rho$ forces disagreement]
\label{lem:rho-outer-outputs}
Under the restriction $\rho$, the outer multiplier strips at base $K$ are forced to output the binary representation of $m$ in $(\vv{x}\vv{y})\vv{z}$ and of $m+1$ in $\vv{x}(\vv{y}\vv{z})$. In particular, the output bits at column $K$ disagree.
\end{lemma}
\begin{proof}
We claim that the $g$-column $K$ of the multiplier $(\vv{x}\vv{y})\vv{z}$ contains exactly $m$ ones, one for each $v\in V$. Indeed, for each $v\in V$, the restriction $\rho$ sets the output bit of $\vv{x}\vv{y}$ at $\col_{\vv{xy}}(v)$ and the input bit
$z_{k(v)}$ to $1$. Their product is a partial product lying in column $\col_{\vv{xy}}(v)+k(v)=K$, and by
Lemma~\ref{lem:indices}(ii),(iii) it is the only partial product in column $K$ containing $z_{k(v)}$. Every other
partial product in column $K$ is forced to $0$, since $\rho$ restricts every other bit of $\vv{z}$ to $0$. By Lemma~\ref{lem:rho-g-sparse}(ii), $\rho$ is a $g$-sparse restriction of the inputs of the multiplier $(\vv{x}\vv{y})\vv{z}$, hence Observation~\ref{obs:strip-counts} applies and the strip at base $K$ is forced to output the binary
representation of $m$ (see Figure~\ref{fig:outer-multiplier}). Symmetrically, column $K$ of the multiplier $\vv{x}(\vv{y}\vv{z})$ contains a $1$ for each $u\in U$, so its strip at base $K$ is forced to output the binary representation of $m+1$. Since $m$ and $m+1$ disagree in their least significant bit, the output bits at column $K$ disagree.
\end{proof}

\section{Resolution lower bound for \texorpdfstring{$\Assoc_n$}{Assoc\_n}}
\label{sec:reduction}

With the restriction $\rho$ constructed (Definition~\ref{def:rho}), we now carry out the substitute and compose steps outlined in Section~\ref{sec:outline-reduction} to prove the lower bound.
Throughout this section, we assume that $n\ge n_0$, and that $\Assoc_n$ uses a strip-local multiplier encoding $\Mul_n$.

For the following definitions, recall that under the restriction $\rho$, every variable of $\Assoc_n$ computes a Boolean function of edge variables $X_e$.

\begin{definition}[Propagated assignment]
\label{def:propagation}
For an assignment $\alpha$ to the edge variables $X_e$, the \emph{propagated assignment} $\bar\alpha$ extends $\rho$ by setting each
variable of $\Assoc_n\!\upharpoonright_\rho$ to the value of its Boolean function at $\alpha$.
\end{definition}

\begin{definition}[Edge- and star-local]
\label{def:edge-star-local}
Under $\rho$, a circuit wire is \emph{edge-local} if it is not constant and computes a Boolean function of a single edge variable
$X_e$. It is \emph{star-local at $w$} if it is neither constant nor edge-local and computes a Boolean function of the variables
$\{X_e:e\in E(w)\}$.
\end{definition}

\begin{definition}[Support of a restricted clause]
\label{def:support}
The \emph{support} of a clause $C$ of $\Assoc_n\!\upharpoonright_\rho$ is the set $\supp(C)$ of vertices $w\in U\cup V$ such that some variable of $C$ depends on an edge variable of the star $E(w)$.
\end{definition}

\begin{lemma}[Locality]
\label{lem:locality}
Under the restriction $\rho$, every wire of the four multipliers and every miter variable is constant, edge-local, or star-local. Moreover, every wire of the inner multiplier $\vv{x}\vv{y}$ computes a Boolean function of a single star $E(v)$, $v\in V$, and every wire of $\vv{y}\vv{z}$ computes a Boolean function of a single star $E(u)$, $u\in U$.
\end{lemma}

\begin{proof}
Since $\rho$ is a $g$-sparse restriction of the inputs for each of the four
multipliers (Lemma~\ref{lem:rho-g-sparse}(ii)) and $\Mul_n$ is a strip-local encoding, under $\rho$ every wire of these multipliers computes a Boolean function of the partial products within a single
strip's base column. The statement then follows for the inner multipliers $\vv{x}\vv{y}$ and $\vv{y}\vv{z}$ by Lemma~\ref{lem:rho-g-sparse}(i).

For the outer multipliers, each of their input bits is either constant or computes a single edge variable $X_e$. Consider, for example, the multiplier $(\vv{x}\vv{y})\vv{z}$. Its input bits $\vv{z}$ are all set to constants by $\rho$. In the other input bit-vector $(\vv{x}\vv{y})$, each unrestricted bit is the base output bit of an inner multiplier strip holding a single partial product
$X_e$ (Definition~\ref{def:rho}(b), Lemma~\ref{lem:rho-g-sparse}(i)). By Observation~\ref{obs:strip-counts} this output bit computes $X_e$.

By Definition~\ref{def:rho}(a),(b), the partial products not forced to $0$ in either outer multiplier are indexed by triples $(u,e,v) \in U \times E \times V$. By Lemma~\ref{lem:indices}(iii), the column of such a partial product depends only on its difference pair
$\bigl(i(u)-i(u_e),\,k(v)-k(v_e)\bigr)$, and each nonzero difference is realized by a unique ordered pair of vertices from $U$ or $V$.
There are therefore four kinds of outer $g$-columns to check, according to which entries of the difference pair are $0$ (Figures~\ref{fig:outer-multiplier} and~\ref{fig:outer-star-local}). It suffices to show that each kind of column only contains edge variables from at most one star.
\begin{itemize}
\item \textbf{Difference pair $(0,0)$.} By Lemma~\ref{lem:indices}(iii)(a), this is the $g$-column $K$, which contains only constant
partial products.
\item \textbf{Difference pair $(d,0)$ with $d\ne0$.} By Lemma~\ref{lem:indices}(iii)(b), we have $d=i(u)-i(u_e)$ for a unique ordered pair
$(u,u_e)$, so the triples mapping to this column are $(u,e,v_e)$ with $e\in E(u_e)$. In the outer multiplier $(\vv{x}\vv{y})\vv{z}$,
this column therefore contains exactly the edge variables from the single star $E(u_e)$. In $\vv{x}(\vv{y}\vv{z})$, the only partial product of this column not forced to $0$ is the product of the input bit $x_{i(u)}$
with the output bit of $\vv{y}\vv{z}$ at $\col_{\vv{yz}}(u_e)$, and $\rho$ sets both bits to $1$.
\item \textbf{Difference pair $(0,d)$ with $d\ne0$.} Symmetrically, in the multiplier $\vv{x}(\vv{y}\vv{z})$ this column contains exactly the edge variables from the single star $E(v_e)$, while in the multiplier $(\vv{x}\vv{y})\vv{z}$ the only partial product not forced to $0$ in this column is the constant $1$.
\item \textbf{Both differences nonzero.} By Lemma~\ref{lem:indices}(iii)(b), the two differences determine ordered pairs $(u,u')$ and $(v,v')$,
so every triple mapping to this column has the form $(u,e,v)$ with $e=(u',v')$. If the graph $H$ contains the edge $(u',v')$, then in both outer multipliers this column contains the single edge variable $X_{(u',v')}$. Otherwise the column contains no edge variables.
\end{itemize}

Finally, we check the miter variables. Each miter variable $e_\ell$ depends only on the partial products of the two outer multipliers in the strip containing column $\ell$. In each of the four cases, these partial products involve edge variables from at most one star, so each miter variable $e_\ell$ is also constant, edge-local, or star-local.
\end{proof}

The next lemma is not used directly in the reduction, but it is needed in the proof of Lemma~\ref{lem:local-soundness}.

\begin{lemma}[Inner multipliers impose perfect matching]
\label{lem:inner-pm}
Under the restriction $\rho$, the inner multipliers impose exactly the perfect-matching constraints on the variables $X_e$: for every assignment $\alpha$ to the edge variables $X_e$,
the propagated assignment $\bar\alpha$ (Definition~\ref{def:propagation}) satisfies $\vv{x}\vv{y}\!\upharpoonright_{\rho}$ if and only if $\alpha$ sets
exactly one edge variable in the star $E(v)$ to true for every $v\in V$. Symmetrically, $\bar\alpha$ satisfies
$\vv{y}\vv{z}\!\upharpoonright_{\rho}$ if and only if $\alpha$ sets exactly one edge variable in the star $E(u)$ to true for every $u\in U$.
\end{lemma}

\begin{proof}
Consider the multiplier $\vv{x}\vv{y}$; the argument for $\vv{y}\vv{z}$ is symmetric.

Suppose first that the assignment $\alpha$ selects exactly one edge in every star $E(v)$, $v\in V$. By Lemma~\ref{lem:rho-g-sparse}(i), for each $v\in V$ the $g$-column $\col_{\vv{xy}}(v)$ contains exactly the partial products
$\{X_e : e\in E(v)\}$. Since $\rho$ is $g$-sparse (Lemma~\ref{lem:rho-g-sparse}(ii)), Observation~\ref{obs:strip-counts} applies, and the strip based at this column outputs $1$ in binary, agreeing with $\rho$. In a strip holding a single partial product $X_e$, the
output bit at the base column computes $X_e$ and is left unset by $\rho$, while the $g-1$ higher output bits of the strip compute $0$. Strips whose
partial products are all forced to $0$ output $0$. In all cases the propagated values agree with $\rho$, so $\bar\alpha$ satisfies every clause of $\vv{x}\vv{y}\!\upharpoonright_{\rho}$.

Conversely, suppose that the propagated assignment $\bar\alpha$ satisfies $\vv{x}\vv{y}\!\upharpoonright_{\rho}$. Since $\bar\alpha$ extends $\rho$, it satisfies the unrestricted encoding
$\vv{x}\vv{y}$, and is therefore a correct evaluation of the multiplier (Definition~\ref{def:exact-encoding}). By Lemma~\ref{lem:rho-g-sparse}(i), for each $v\in V$ the strip at base $\col_{\vv{xy}}(v)$ outputs the binary count of the true edge variables in $E(v)$. Since the restriction $\rho$ sets this strip's output to $1$ in binary, exactly one edge variable of
$E(v)$ is true in the edge assignment $\alpha$.
\end{proof}

The reduction will use the following lemma to derive the clauses of $\Assoc_n\!\upharpoonright_\rho$ from $\SPM(H)$.

\begin{lemma}[Local soundness]
\label{lem:local-soundness}
Let $C$ be a clause of $\Assoc_n\!\upharpoonright_\rho$. If an assignment $\alpha$ to the edge variables satisfies the clauses $\operatorname{ExactOne}(X_{E(w)})$ for every $w\in\supp(C)$, then the propagated assignment $\bar\alpha$ (Definition~\ref{def:propagation}) satisfies $C$.
\end{lemma}

\begin{proof}
The clause $C$ originates from one of three components: the inner multipliers $\vv{x}\vv{y}$ and $\vv{y}\vv{z}$, the outer multipliers $(\vv{x}\vv{y})\vv{z}$ and $\vv{x}(\vv{y}\vv{z})$,
or the clauses defining the miter variables $e_i$. The restricted formula $\Assoc_n\!\upharpoonright_\rho$ does not contain the long miter clause $e_0\vee\cdots\vee e_{4n-1}$ since $\rho$ satisfies it by assigning $e_K=1$. We check the cases in turn.

\emph{Inner multipliers.}
Suppose that $C$ originates from the inner multiplier $\vv{x}\vv{y}$; the case of $\vv{y}\vv{z}$ is symmetric.
Let $\alpha'$ be an assignment to the edge variables that agrees with $\alpha$ on the stars $E(v)$ with $v\in V\cap\supp(C)$, and that selects exactly one edge in every remaining star $E(v)$, where $v\in V\setminus\supp(C)$. This is well defined because the stars $E(v)$ partition the edges $E$. Since $\alpha$ satisfies $\operatorname{ExactOne}(X_{E(v)})$ for every $v\in V\cap\supp(C)$, the assignment $\alpha'$ in fact selects exactly one edge in every star $E(v)$, $v\in V$. By Lemma~\ref{lem:inner-pm}, the propagated assignment $\bar\alpha'$ satisfies every clause of $\vv{x}\vv{y}\!\upharpoonright_{\rho}$, and in particular satisfies $C$.

Since every variable of $C$ depends only on edge variables $X_e$ with an endpoint $v_e\in\supp(C)$, on which $\alpha'$ is defined to agree with $\alpha$, the propagated assignments $\bar\alpha$ and $\bar\alpha'$ agree on the variables of $C$. Hence $\bar\alpha$ satisfies $C$ as well.

\emph{Outer multipliers.} The restriction $\rho$ assigns no wires in the outer multipliers; it fixes only some of their inputs, namely bits of $\vv{x}$, $\vv{z}$, $\vv{x}\vv{y}$, and $\vv{y}\vv{z}$. The propagated assignment $\bar\alpha$ therefore satisfies every clause of the outer multipliers.

\emph{Miter clauses.} The miter variables $e_\ell$ with $\ell\ne K$ are unassigned by $\rho$, and their defining clauses are gate definitions, so the
propagated assignment $\bar\alpha$ satisfies them as well.

It remains to check the clauses defining $e_K$. As $\rho$ sets $e_K=1$, these clauses assert that the two outer multipliers disagree at output bit $K$. By Lemma~\ref{lem:rho-outer-outputs}, these two output bits are forced to disagree under $\rho$, so $\bar\alpha$ satisfies these clauses.
\end{proof}

\begin{lemma}[Reduction from perfect matching]
\label{lem:reduction}
From any resolution refutation of $\Assoc_n$ one can obtain a refutation of
the star-local extension $\SPM(H)$ (Definition~\ref{def:star-local-ext}) with only a constant-factor increase in size.
\end{lemma}

\begin{proof}
By the substitution lemma (Lemma~\ref{lem:substitution}), a size $s$ refutation of $\Assoc_n$ yields a refutation of
$\Assoc_n\!\upharpoonright_\rho$ of size at most $s$.

Next, define a substitution $\sigma$ from the variables of $\Assoc_n\!\upharpoonright_\rho$ to constants and literals over the variables of
$\SPM(H)$ as follows. Every wire that is forced to a constant under $\rho$ maps to that constant. The unset input bits
$y_{j(e)}$ map to the edge variables $X_e$. Each edge-local wire maps to the literal it computes, $X_e$ or $\neg X_e$. Every remaining wire is, by Lemma~\ref{lem:locality}, a Boolean function $f$ of a single
star $E(w)$, and maps to the extension variable $Z_{w,f}$. Wires computing the same function of the same star all map to the same variable $Z_{w,f}$. By the substitution lemma (Lemma~\ref{lem:substitution}), the refutation of $\Assoc_n\!\upharpoonright_\rho$ yields a refutation of
$(\Assoc_n\!\upharpoonright_\rho)^\sigma$ of size at most $s$.

It remains to derive the clauses of the restricted, substituted formula $(\Assoc_n\!\upharpoonright_\rho)^\sigma$ from $\SPM(H)$. Each clause $C^\sigma$ of $(\Assoc_n\!\upharpoonright_\rho)^\sigma$ is the result of applying the substitution $\sigma$ to a clause $C$ of the restricted formula $\Assoc_n\!\upharpoonright_\rho$. The clause $C$ has constant width: all clauses of $\Assoc_n$ have constant width except the long miter clause $e_0\vee\cdots\vee e_{4n-1}$, which $\rho$ satisfies. By Lemma~\ref{lem:locality}, each variable of $C$ depends on the edge variables of at most one star, so $\supp(C)$ has constant size. 

We claim that the clauses $\operatorname{ExactOne}(X_{E(w)})$ and $\operatorname{Def}(Z_{w,f})$ for $w\in \supp(C)$ imply $C^\sigma$. Indeed, let $\alpha$ be an assignment to the variables of $\SPM(H)$ that satisfies these clauses. By Lemma~\ref{lem:local-soundness}, the propagated assignment $\bar\alpha$ (Definition~\ref{def:propagation}) satisfies $C$. Furthermore, the definitions $\operatorname{Def}(Z_{w,f})$ ensure that the substituted values agree with the propagated values, i.e., $\alpha(\sigma(q))=\bar\alpha(q)$ for every variable $q$ of $C$. Hence, $\alpha$ satisfies $C^\sigma$. By Lemma~\ref{lem:star-local-derivations}, some subclause of $C^\sigma$ therefore has a constant-size derivation from $\SPM(H)$. Applying the refutation of $(\Assoc_n\!\upharpoonright_\rho)^\sigma$ to these subclauses gives a refutation of $\SPM(H)$ of size $O(s)$.
\end{proof}

\begin{theorem}[Lower bound for multiplier associativity]
\label{thm:main}
Let $\Mul_n$ be a strip-local multiplier encoding. For all sufficiently large $n$,
\[
\sizeRes(\Assoc_n)\ge 2^{\Omega((n/\log n)^{1/4})}.
\]
\end{theorem}

\begin{proof}
Choose
\[
m=\left\lfloor c_0(n/\log n)^{1/4}\right\rfloor
\]
for a sufficiently small absolute constant $c_0$. Theorem~\ref{thm:iss-graphs} provides a bipartite expander $H=(U,V,E)$ with $|U|=m+1$ and $|V|=m$ whose maximum degree is at most a constant $\Delta$, and no isolated vertices, as Section~\ref{sec:hard-source} requires. Theorem~\ref{thm:star-local-hard} then shows that every resolution refutation of $\SPM(H)$ has size $2^{\Omega(m)}$. 
Lemma~\ref{lem:indices} gives $n_0=O(m^4\log m)$, so with a sufficiently small $c_0$ its construction fits into the $n$
input bits; unused high input bits are set to zero.

By Lemma~\ref{lem:reduction}, a refutation of $\Assoc_n$ of size $s$ gives a refutation of $\SPM(H)$ of size $O(s)$. Hence $s\ge 2^{\Omega(m)}=2^{\Omega((n/\log n)^{1/4})}$.
\end{proof}

\begin{remark}[Truncated multiplication]
\label{rem:truncated}
The above lower bound, as well as the bounded-depth Frege lower bound we will prove in Section~\ref{sec:frege}, applies unchanged to associativity of truncated multiplication, where each multiplier outputs only the $n$ low-order bits of
the product and the miter compares $n$ bits. The construction places the
input positions $i(u), j(e), k(v)$, the restricted strip columns, and the
miter variable $e_K$ all below $n_0 \le n$. So under $\rho$, every partial product
not forced to $0$ lies in a column below $n_0$.
\end{remark}

\section{Bounded-depth Frege lower bound for \texorpdfstring{$\Assoc_n$}{Assoc\_n}}
\label{sec:frege}

In this section we prove that $\Assoc_n$ is also hard for bounded-depth Frege systems. We reuse the reduction of Section~\ref{sec:reduction} with the odd grid in place of the expander $H$. The source of hardness is now H\r{a}stad's lower bound for bounded-depth Frege refutations of the corresponding perfect-matching principle~\cite{hastad-php}. Since Frege lines are formulas rather than clauses, we can directly substitute a constant-size formula for each wire, so the extension variables $Z_{w,f}$ from Section~\ref{sec:hard-source} are no longer needed.

\subsection{Bounded-depth Frege systems}
\label{sec:frege-prelims}

\paragraph{Formulas and depth.}
We follow the conventions of H\r{a}stad~\cite{hastad-php}. Formulas are built from variables using negation and disjunctions of unbounded arity. Conjunctions are expressed as $\neg(\neg F_1\vee\cdots\vee\neg F_k)$. The \emph{depth} of a formula is the maximum number of alternations between $\vee$ and $\neg$ on a path from its root to a leaf. Thus a literal has depth $0$ and a clause has depth at most $1$. The \emph{size} of a formula is its number of symbols.

\paragraph{Frege systems.}
A \emph{Frege system} is specified by a finite set of inference rules. Each rule has the form
\[
\frac{F_1(p_1,\ldots,p_k)\quad\cdots\quad F_j(p_1,\ldots,p_k)}{F(p_1,\ldots,p_k)}
\]
for fixed formulas $F_1,\ldots,F_j$ and $F$ over the variables $p_1,\ldots,p_k$, and may be applied with arbitrary formulas substituted for $p_1,\ldots,p_k$. A \emph{derivation} of a formula $G$ from a set of formulas $\Gamma$ is a sequence of formulas ending in $G$, each of which belongs to $\Gamma$ or follows from earlier formulas by an inference rule. The rules are required to be sound and \emph{implicationally complete}, meaning that $\Gamma$ implies $G$ if and only if $G$ has a derivation from $\Gamma$~\cite{cook-reckhow}. A \emph{refutation} of a set of formulas $\Gamma$ is a derivation of $\bot$ from $\Gamma$. The \emph{depth} of a derivation is the maximum depth of its formulas, and its \emph{size} is the total size of its formulas. We write $\sizeF{d}(\Phi)$ for the minimum size of a Frege refutation of a set of formulas $\Phi$ of depth at most $d$.

We fix an arbitrary Frege system throughout this section. H\r{a}stad proved his lower bound for one specific system~\cite[Section~2.2]{hastad-php}, but noted that the argument uses only that the rules are sound and of constant size, so his lower bound holds for every Frege system.

\paragraph{Formula substitutions.}
A \emph{formula substitution} $\sigma$ maps each variable to a formula. We write $F^\sigma$ or $\pi^\sigma$ for the result of applying $\sigma$ to a formula $F$, or a derivation $\pi$, respectively. The \emph{depth} and \emph{size} of a substitution $\sigma$ are the maximum depth and size of the formulas $\sigma(x)$. We write $\operatorname{depth}(\cdot)$ and $|\cdot|$ for the depth and size of a formula, a derivation, or a substitution. Formula substitutions satisfy the following analogue of the substitution lemma for resolution (Lemma~\ref{lem:substitution}).

\begin{lemma}[{Frege substitution lemma, cf.~\cite[Lemma~2.5]{cook-reckhow}}]
\label{lem:frege-substitution}
If $\pi$ is a derivation of $G$ from $\Gamma$, then $\pi^\sigma$ is a derivation of $G^\sigma$ from $\Gamma^\sigma$ of depth at most $\operatorname{depth}(\pi)+\operatorname{depth}(\sigma)+1$ and size at most $|\sigma|\cdot|\pi|$.
\end{lemma}

\begin{proof}
Frege rules are closed under substitution, so each inference of $\pi$ is still an inference of $\pi^\sigma$. The depth grows by at most $\operatorname{depth}(\sigma)$ plus one for a possible alternation where $\sigma(x)$ is substituted.
\end{proof}

The following lemma is the Frege system analogue of Lemma~\ref{lem:star-local-derivations}, and follows immediately from implicational completeness.

\begin{lemma}[Local consequences have constant-size Frege derivations]
\label{lem:frege-local-derivations}
If a set $\Gamma$ of clauses over $k$ variables implies a formula $G$, then $G$ has a derivation from $\Gamma$ whose depth and size depend only on $k$ and $|G|$.
\end{lemma}

\subsection{The odd grid}
\label{sec:grid}

For an odd integer $t\ge3$, the \emph{grid} $\Grid_t$ is the graph with vertex set $[1,t]^2$ in which two vertices are adjacent if they agree in one coordinate and differ by $1$ in the other. Color the vertices as a chessboard, with a vertex $(i,j)$ white if $i+j$ is even and black otherwise. Every edge joins a white vertex to a black vertex, so $\Grid_t$ is bipartite with sides
\[
U=\{\text{white vertices}\},\qquad V=\{\text{black vertices}\},\qquad |U|=\tfrac{t^2+1}{2},\qquad |V|=\tfrac{t^2-1}{2},
\]
and every vertex has degree at most $4$. Thus $\Grid_t$ satisfies the standing assumptions of Section~\ref{sec:hard-source}: it is a bipartite graph $H=(U,V,E)$ with $|U|=m+1$ and $|V|=m$, where $m=(t^2-1)/2$, and has maximum degree at most $\Delta=4$. We may therefore apply the reduction of Section~\ref{sec:reduction} with $H=\Grid_t$.

H\r{a}stad~\cite{hastad-php} studied the CNF $\PM(\Grid_t)$ under the name of the \emph{functional and onto pigeonhole principle} on the odd grid. From this perspective the white vertices are the pigeons and the black vertices are the holes.

\begin{theorem}[{H\r{a}stad~\cite[Theorem~5.1]{hastad-php}}]
\label{thm:hastad}
Let $t$ be an odd integer and let \mbox{$d\le O(\log t/\log\log t)$}. Then any depth $d$ Frege refutation of $\PM(\Grid_t)$ requires size
\[
\exp\bigl(\Omega\bigl(t^{1/(2d-1)}(\log t)^{O(1)}\bigr)\bigr).
\]
In particular, there is a constant $c>0$ such that for every odd $t$ and every $d\le c\log t/\log\log t$,
\[
\sizeF{d}(\PM(\Grid_t))\ \ge\ \exp\bigl(t^{c/d}\bigr).
\]
\end{theorem}

\subsection{Reduction and lower bound}
\label{sec:frege-transfer}

We now carry out the reduction of Section~\ref{sec:reduction} in a bounded-depth Frege system. As in Section~\ref{sec:reduction}, we assume that $n\ge n_0$, and $\Assoc_n$ uses a strip-local multiplier encoding $\Mul_n$.

\begin{lemma}[Reduction from $\PM(H)$]
\label{lem:frege-transfer}
There are constants $d_0$ and $s_0$ such that for every $d$, if $\pi$ is a depth $d$ Frege refutation of $\Assoc_n$, then $\PM(H)$ has a depth $d+d_0$ Frege refutation of size at most $s_0|\pi|$.
\end{lemma}

\begin{proof}
This reduction follows the same three steps as the resolution reduction (Lemma~\ref{lem:reduction}). The main difference is that here, in order to use Lemma~\ref{lem:frege-substitution}, we apply $\rho$ as a formula substitution, which does not simplify away constants as a restriction would (otherwise an inference may become invalid). In order to reuse the lemmas from Section~\ref{sec:reduction}, we will therefore need to derive the formulas of $\Assoc_n^\rho$ (where $\rho$ is applied as a formula substitution) from the clauses of $\Assoc_n\!\upharpoonright_\rho$ (where $\rho$ is applied as a restriction). The substitution and local-derivation steps are essentially unchanged, except that formulas take the place of extension variables.

\emph{Application of $\rho$.} 
Let $\pi$ be a depth $d$ Frege refutation of $\Assoc_n$ of size $s$. By Lemma~\ref{lem:frege-substitution}, applying $\rho$ as a formula substitution gives a depth $d+O(1)$, size $O(s)$ refutation $\pi^\rho$ of $\Assoc_n^\rho$. Aside from the long miter clause $e_0\vee\cdots\vee e_{4n-1}$, each clause $C^\rho$ of $\Assoc_n^\rho$ has constant size, and hence has a derivation of constant depth and size (Lemma~\ref{lem:frege-local-derivations}) from its counterpart $C\!\upharpoonright_\rho$ in $\Assoc_n\!\upharpoonright_\rho$ (if the clause $C$ is satisfied by $\rho$ then we instead derive $C^\rho$ by weakening from $\neg\bot$). The long miter clause after applying the formula substitution $\rho$ becomes $e_0\vee\cdots\vee e_{K-1} \vee \neg \bot \vee e_{K+1} \vee \cdots \vee e_{4n-1}$. We derive this in constant depth and size $O(n)$ as follows: fix a constant-size derivation of $p \vee q$ from $p$, then substitute $\neg \bot$ for $p$ and $e_0\vee\cdots\vee e_{K-1} \vee e_{K+1} \vee \cdots \vee e_{4n-1}$ for $q$.

Prepending these derivations to $\pi^\rho$ gives a refutation $\pi\!\upharpoonright_\rho$ of $\Assoc_n\!\upharpoonright_\rho$ of depth $d+O(1)$ and size $O(s)$.

\emph{DNF substitution.} By Lemma~\ref{lem:locality}, every variable of $\Assoc_n\!\upharpoonright_\rho$ is an edge variable or a wire that is constant, edge-local, or star-local. Define a formula substitution $\sigma$ accordingly: an edge variable $X_e$ maps to itself; a constant wire maps to the corresponding value $\neg\bot$ or $\bot$; an edge-local wire maps to the corresponding literal $X_e$ or $\neg X_e$; and a wire that is star-local at $w$ and computes $f(X_{E(w)})$ maps to the canonical DNF of $f$, which has depth at most $3$ and constant size. Then $\sigma$ is a substitution with depth at most $3$ and constant size. By Lemma~\ref{lem:frege-substitution}, $(\pi\!\upharpoonright_\rho)^\sigma$ is a depth $d+O(1)$, size $O(s)$ refutation of $(\Assoc_n\!\upharpoonright_\rho)^\sigma$.

\emph{Local derivations.} 
It remains to derive the formulas of $(\Assoc_n\!\upharpoonright_\rho)^\sigma$ from $\PM(H)$. Each such formula is $C^\sigma$ for a clause $C$ of $\Assoc_n\!\upharpoonright_\rho$. The clause $C$ has constant width, since $\rho$ satisfies the long miter clause, and by Lemma~\ref{lem:locality} each variable of $C$ depends on the edge variables of at most one star, so $\supp(C)$ has constant size. We claim that the clauses $\operatorname{ExactOne}(X_{E(w)})$ for $w\in\supp(C)$ imply $C^\sigma$. Indeed, let $\alpha$ be an assignment to the edge variables that satisfies these clauses. By Lemma~\ref{lem:local-soundness}, the propagated assignment $\bar\alpha$ (Definition~\ref{def:propagation}) satisfies the clause $C$. Furthermore, $\sigma$ maps each variable $q$ of clause $C$ to a formula computing its Boolean function, so the substituted values agree with the propagated values, i.e., $\alpha(\sigma(q))=\bar\alpha(q)$. Hence $\alpha$ satisfies $C^\sigma$. By Lemma~\ref{lem:frege-local-derivations}, the formula $C^\sigma$ therefore has a constant-depth, constant-size derivation from $\PM(H)$. Applying the refutation $(\pi\!\upharpoonright_\rho)^\sigma$ of $(\Assoc_n\!\upharpoonright_\rho)^\sigma$ to these derived formulas gives a refutation of $\PM(H)$ with depth $d+O(1)$ and size $O(s)$.
\end{proof}

\begin{theorem}[Bounded-depth Frege lower bound]
\label{thm:frege}
Let $\Mul_n$ be a strip-local multiplier encoding. There is a constant $\gamma>0$, depending only on the encoding, such that for all sufficiently large $n$ and all $d\le\gamma\log n/\log\log n$,
\[
\sizeF{d}(\Assoc_n)\ \ge\ \exp\bigl(n^{\gamma/d}\bigr).
\]
\end{theorem}

\begin{proof}
Lemma~\ref{lem:indices} gives a value $n_0=O(m^4\log m)=O(t^8 \log t)$. Let $n_0(t)$ denote this value as a function of $t$.
Set $t$ to the largest odd integer such that $n_0(t)\le n$. Then $n_0(t+2) > n$, therefore $t=\Omega((n/\log n)^{1/8})$. So when $n$ is sufficiently large, we have $t \ge n^{1/9}$. 

Let $c$ be the constant of H\r{a}stad's lower bound (Theorem~\ref{thm:hastad}), and let $d_0$ and $s_0$ be the constants of Lemma~\ref{lem:frege-transfer}. Set the constant $\gamma=c/(18(1+d_0))$.

Now let $d\le\gamma\log n/\log\log n$, and suppose that $\Assoc_n$ has a depth $d$ Frege refutation of size $s$. 
Applying the reduction of Lemma~\ref{lem:frege-transfer} yields a Frege refutation of $\PM(\Grid_t)$ of depth $d+d_0$ and size at most $s_0s$. 
Since $\gamma\leq c/18$ and $t\ge n^{1/9}$, we have $d+d_0\le c\log t/\log\log t$ for all sufficiently large $n$. 
Therefore we may apply Theorem~\ref{thm:hastad} with depth $d+d_0 \le (1+d_0)d$, which gives
\[
s_0s\ \ge\ \exp\bigl(t^{c/((1+d_0)d)}\bigr)\ \ge\ \exp\bigl(n^{c/(9(1+d_0)d)}\bigr)\ =\ \exp\bigl(n^{2\gamma/d}\bigr).
\]
Finally, since $d\le\gamma\log n/\log\log n$, the exponent $n^{2\gamma/d}$ exceeds $n^{\gamma/d}$ by at least a factor of $\log n$, which absorbs $s_0$ for all sufficiently large $n$. 
Hence $s\ge\exp(n^{\gamma/d})$.
\end{proof}

\begin{remark}
\label{rem:frege-vs-res}
A resolution refutation is a depth $1$ Frege refutation in a Frege system whose rules include resolution, so Theorem~\ref{thm:frege} with $d=1$ yields a resolution lower bound of $\exp(n^{\gamma})$. However, the proof gives $\gamma<1/18$, so this bound is weaker than the lower bound of $2^{\Omega((n/\log n)^{1/4})}$ given by Theorem~\ref{thm:main}.
\end{remark}

\begin{remark}[Number of lines]
\label{rem:frege-lines}
H\r{a}stad also proved that for $d\le O(\log t/\log\log t)$, if every line of a depth $d$ Frege refutation of $\PM(\Grid_t)$ has size at most $M$, then the refutation has at least $\exp\bigl(\Omega\bigl(t/((\log M)^{2d-1}(\log t)^{O(1)})\bigr)\bigr)$ lines~\cite[Theorem~7.1]{hastad-php}. The proof of Lemma~\ref{lem:frege-transfer} turns a depth $d$ Frege refutation of $\Assoc_n$ with $L$ lines, each of size at most $M$, into a depth $d+d_0$ Frege refutation of $\PM(\Grid_t)$ with $O(L)$ lines of size $O(M)$, so $L$ satisfies the same bound with $d+d_0$ in place of $d$, $O(M)$ in place of $M$, and $t=\Omega((n/\log n)^{1/8})$. In particular, for constant $d$ and lines of polynomial size, a depth $d$ Frege refutation of $\Assoc_n$ requires $\exp\bigl(n^{1/8-o(1)}\bigr)$ lines.
\end{remark}

\section{Conclusion and open problems}
\label{sec:conclusion}

In this paper we resolve an open problem posed by Beame and Liew in~\cite{beame-liew} by proving an exponential lower bound on the size of resolution refutations of the multiplier associativity formula $\Assoc_n$ (Theorem~\ref{thm:main}). Using the same method, we also obtain an exponential lower bound for bounded-depth Frege systems (Theorem~\ref{thm:frege}). These lower bounds apply to any strip-local multiplier encoding, which includes array and Wallace-tree multipliers. The resolution lower bound closes off a possibility raised in the introduction: for associativity, at least, there are no polynomial size resolution proofs that CDCL solvers merely fail to find.

Beyond the lower bounds themselves, our reduction uncovers a connection between two seemingly unrelated principles: the associativity of grade-school multiplication and perfect matching in bipartite graphs.

An immediate open problem is to improve the resolution lower bound quantitatively. The
exponent $(n/\log n)^{1/4}$ in Theorem~\ref{thm:main}
reflects the cost of our embedding (Lemma~\ref{lem:indices}), which requires $O(m^4 \log m)$ input
bits to embed an $m$-vertex matching instance; a denser embedding could sharpen
the exponent. Moreover, the constants in our bound do not give concrete hardness at practical bit-widths such as $n = 64$. In this respect the pigeonhole principle is similar, since the proven bounds~\cite{haken,beame-pitassi} are vacuous at the sizes where solvers actually fail. A lower bound of this kind nevertheless addresses the empirical hardness at the right level: it shows that exponential scaling is unavoidable for resolution, and hence for CDCL solvers, no matter how good their heuristics. We leave
open whether $\Assoc_n$ requires resolution refutations of size
$2^{\Omega(n)}$, and whether the hardness can be made concrete at practical bit-widths.

It also remains open to determine the proof complexity of multiplier associativity in polynomial calculus (with or without resolution) and in cutting planes. Our reduction yields no lower bound in either case, since the perfect-matching principle has small proofs in both. Nevertheless, we conjecture that (bit-level) multiplier associativity remains hard for these systems as well.

Another open line of work is to prove lower bounds for the associativity of other multiplier architectures. Our lower bounds do not apply to Booth multipliers, which are based on partial product summation but are not strip-local, nor to Karatsuba, Toom--Cook, and Sch\"onhage--Strassen multipliers, which are not based on partial product summation.

Finally, we conjecture that the $O(n^6\log n)$-size resolution proofs of commutativity and distributivity found by Beame and Liew~\cite{beame-liew} are asymptotically optimal. A matching lower bound would show that SAT solvers inherently struggle with these identities, since proofs of this size are already prohibitive at practical bit-widths.

\section*{Acknowledgements}

We thank Paul Beame for suggesting a simplified proof of Theorem~\ref{thm:star-local-hard} and for other helpful comments and pointers to the literature, and thank Andrei Constantinescu for suggestions that improved the exposition of Section~\ref{sec:hard-source}.

\appendix
\section{Wallace-tree multipliers are strip-local}
\label{app:architectures}

We first define the carry-lookahead adder.

\begin{definition}[Carry-lookahead adder]
\label{def:cla}
A \emph{carry-lookahead adder} (CLA) uses a tree structure to add two bit-vectors $\vv{x},\vv{y}$ with only logarithmic depth. The
$4$-bit CLA computes, for each pair $x_i,y_i$, the \emph{generate} and \emph{propagate} signals
\[
q_i=x_iy_i,
\qquad
p_i=x_i\oplus y_i.
\]
Then, writing $c_i$ for the carry bit in the $i$-th column, we have
\[
c_{i+1}=q_i\oplus(p_ic_i),
\]
with carry-in $c_0=0$. Unrolling this recurrence gives the following equations, which compute each carry bit in parallel
from the values $q_i,p_i$ and $c_0$:
\begin{align*}
c_1&=q_0\oplus p_0c_0\\
c_2&=q_1\oplus q_0p_1\oplus c_0p_0p_1\\
c_3&=q_2\oplus q_1p_2\oplus q_0p_1p_2\oplus c_0p_0p_1p_2\\
c_4&=q_3\oplus q_2p_3\oplus q_1p_2p_3\oplus q_0p_1p_2p_3\oplus c_0p_0p_1p_2p_3.
\end{align*}
These values are used to compute the outputs $(x+y)_i=c_i\oplus x_i\oplus y_i$. The $4$-bit CLA additionally computes the
\emph{group propagate} and \emph{group generate}
\[
p^1=p_0p_1p_2p_3,
\qquad
q^1=q_3\oplus q_2p_3\oplus q_1p_2p_3\oplus q_0p_1p_2p_3,
\]
where the superscript indicates the level.

We construct a $16$-bit CLA with two levels, whose first half is shown in Figure~\ref{fig:cla}. At the zeroth level we arrange
four $4$-bit CLAs, the $k$-th CLA taking inputs $x_i,y_i$, $i\in[4k,4k+4)$. We denote the $k$-th CLA's group propagate and
generate by $p^1_{4k},q^1_{4k}$. Then the carries $c_4,c_8,c_{12},c_{16}$ can be computed by the equations
\begin{align*}
c_4&=q^1_0\oplus p^1_0c_0\\
c_8&=q^1_4\oplus q^1_0p^1_4\oplus c_0p^1_0p^1_4\\
c_{12}&=q^1_8\oplus q^1_4p^1_8\oplus q^1_0p^1_4p^1_8\oplus c_0p^1_0p^1_4p^1_8\\
c_{16}&=q^1_{12}\oplus q^1_8p^1_{12}\oplus q^1_4p^1_8p^1_{12}\oplus q^1_0p^1_4p^1_8p^1_{12}\oplus c_0p^1_0p^1_4p^1_8p^1_{12}.
\end{align*}
Notice that these equations are isomorphic to the previous equations for computing carries within each $4$-bit CLA. We can reuse
the same circuitry from the $4$-bit CLA to compute these carries, as well as the group propagate and generate for the next level.
We can repeat this process to construct larger CLAs, with each iteration able to handle four times the bit-width.
\end{definition}

\begin{figure}[ht]
\centering
\includegraphics[width=\textwidth]{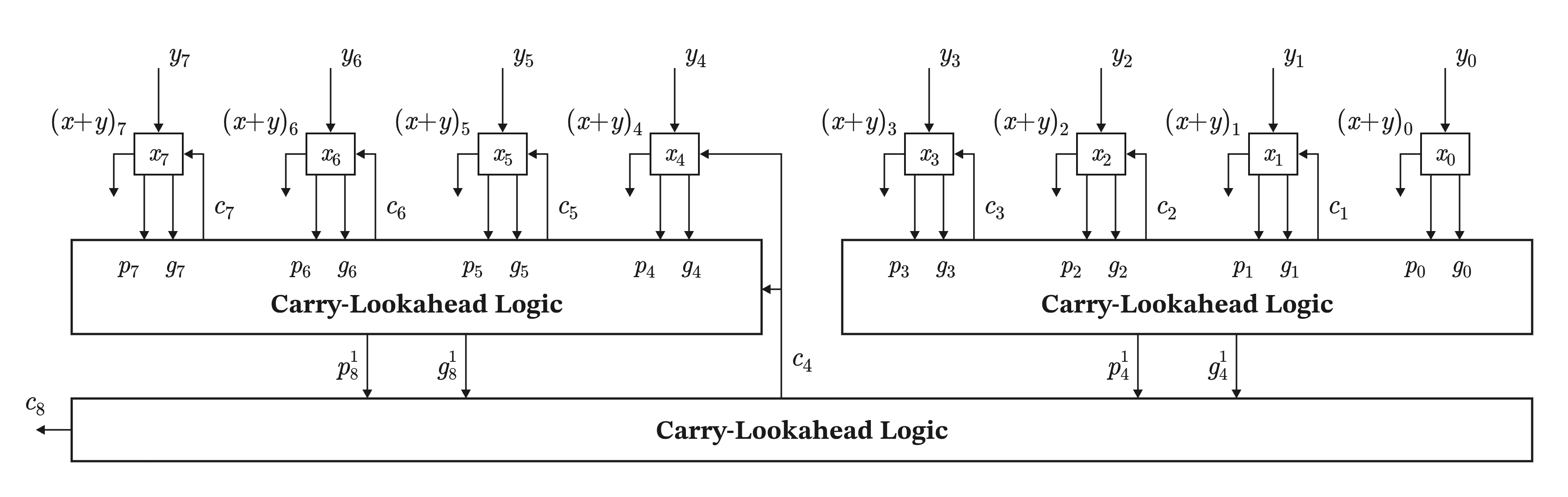}
\caption{An $8$-bit carry-lookahead adder built from two $4$-bit CLAs. Each bit cell combines $x_i$ and $y_i$ into its signal
pair $(p_i,q_i)$ and sum bit. Each CLA computes the carries into its four cells and feeds its group pair, written $(p^1,q^1)$, to the
second-level circuitry, which computes the carries $c_4$ and $c_8$.}
\label{fig:cla}
\end{figure}

We next define the Wallace-tree multiplier. Figure~\ref{fig:wallace-tree} shows a $9$-bit instance.

\begin{definition}[Wallace-tree multiplier]
\label{def:wallace}
A \emph{Wallace-tree multiplier} takes input $(\vv{x},\vv{y})$ and computes each partial product of Definition~\ref{def:partial-products}
by an AND gate. It then goes through $\ell_{\max}<\log_{3/2}n$ layers of parallel adders to reduce the $n$ starting rows
of partial products to two rows of equivalent arithmetic weight.

We define layer variables $t_{\ell,i,j}$, where $\ell$ is the layer, $i$ is the column containing the variable, and $j$ is the row
(note that this differs from the convention for a partial product $t_{i,j}$, whose subscripts are input positions and whose column is
$i+j$). A variable in column $i$ has arithmetic weight $2^i$. We denote the \emph{subcolumn} of column $i$ within layer $\ell$ by
\[
\Col(\ell,i)=\{t_{\ell,i,j}\ \text{for all}\ j\}.
\]
In the zeroth layer, the variables are the partial products, arranged so that $\Col(0,i)$ holds the partial products of column $i$:
\[
t_{0,i,j}=x_{i-j}\wedge y_j\quad\text{for }i<n,
\qquad
t_{0,i,j}=x_{n-1-j}\wedge y_{i-n+j+1}\quad\text{for }i\ge n.
\]

We now specify how to construct layer $\ell+1$ from layer $\ell$. We partition the rows of layer $\ell$ into sets of three, from top to
bottom. Adder $A_{\ell,i,j}$ takes as input the variables in the $i$-th column
of the $j$-th set of three rows, replacing nonexistent variables with the constant $0$. On inputs $a_0,a_1,a_2$ it outputs the sum bit
$s=a_0\oplus a_1\oplus a_2$ and the carry bit $c=\MAJ(a_0,a_1,a_2)$. For each set of rows $j=0,1,\dots$ and each column $i$, we append
adder $A_{\ell,i,j}$'s sum bit to subcolumn $\Col(\ell+1,i)$ and its carry bit to subcolumn $\Col(\ell+1,i+1)$. An adder with no constant
inputs is the full adder of Section~\ref{sec:prelims}; an adder with exactly one constant input is a \emph{half adder}, computing
$s=a_0\oplus a_1$ and $c=a_0\wedge a_1$; an adder with two constant inputs acts as a wire.

Each layer reduces the number of rows from $N$ to $\lceil 2N/3\rceil$, leaving the last layer $\ell_{\max}$ with only two rows. Every adder has the same total weight of inputs and outputs, so each layer has the same weighted sum
$\sum_{i,j}t_{\ell,i,j}\,2^{i}$ as the partial products. Using a
$2n$-bit carry-lookahead adder (Definition~\ref{def:cla}), we sum these two rows, outputting the product in the multiplier's output bits.
\end{definition}

\begin{observation}
\label{obs:wallace-exact}
The Wallace-tree multiplier is an exact multiplier encoding (Definition~\ref{def:exact-encoding}).
\end{observation}

\begin{figure}[p]
\centering
\includegraphics[width=0.85\textwidth]{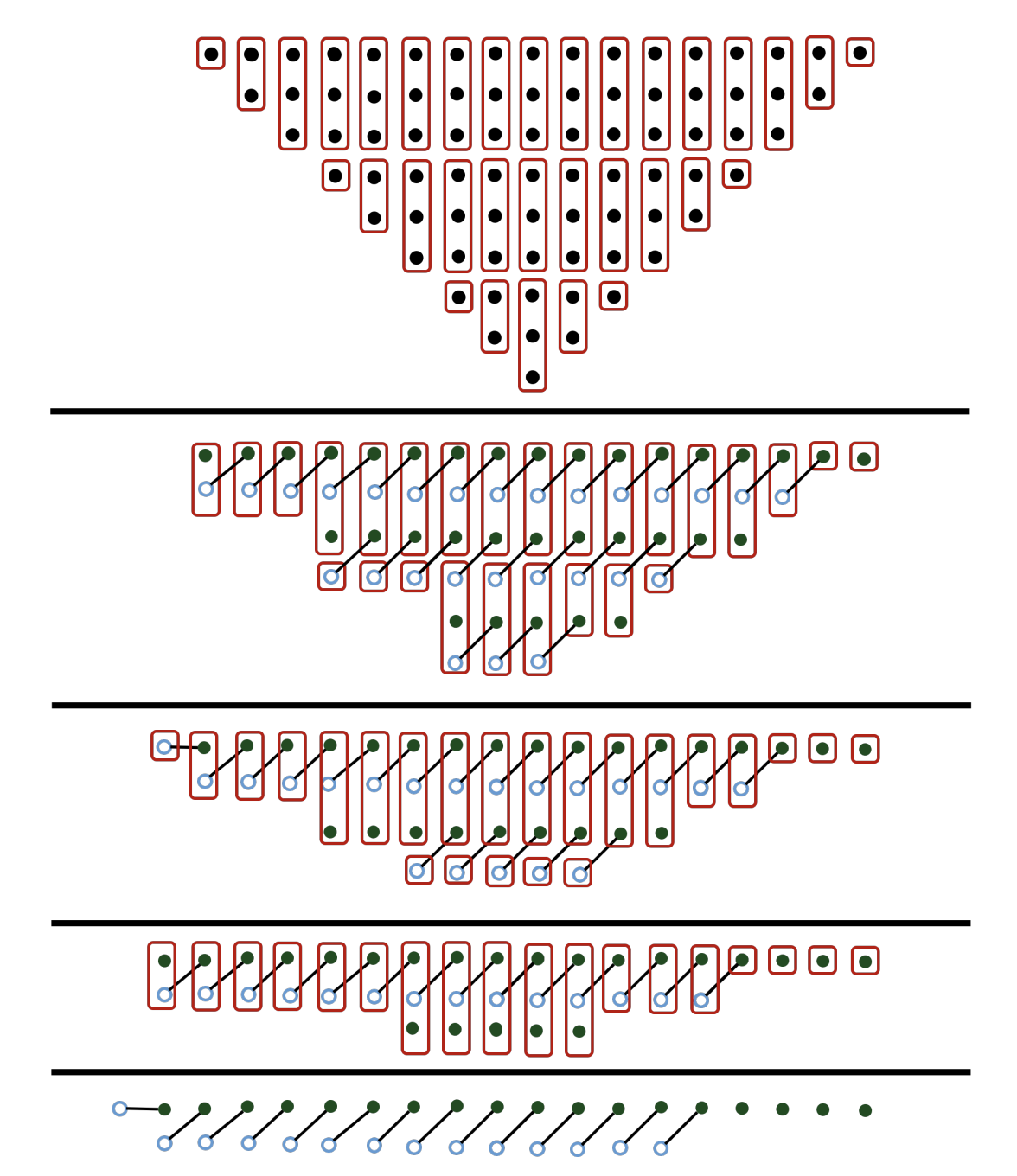}
\caption{Dot diagram for a $9\times9$ Wallace-tree multiplier. Hollow dots represent carry bits and solid dots represent sum bits. Dots
connected by an edge are output by the same adder.}
\label{fig:wallace-tree}
\end{figure}

\begin{lemma}[Wallace-tree strip-locality]
\label{lem:wallace-tree-local}
Fix $g\ge1$ and a $g$-sparse restriction of the input bit-vectors. Then every layer variable of the Wallace tree is constant or a
Boolean function of the partial products of the strip containing it, and every layer variable in the top column of a strip is the
constant $0$.
\end{lemma}

\begin{proof}
In the Wallace tree, the only wires crossing a column boundary are adder carries, so like in the proof of
Proposition~\ref{prop:array-strip-local} it suffices to show that every carry entering a strip is the constant $0$. This is
trivial for the strip at base $0$.

Assume the inductive hypothesis that every
carry entering the strip with base $bg$ is zero. Then every layer variable $t_{\ell,i,j}$ in the
strip is a function of the strip's partial products not forced to $0$. By $g$-sparsity, the strip contains fewer than $2^{g-1}$ such partial products, and they are all contained in the base
column $bg$, so their total weight is less than $2^{g-1}\cdot2^{bg}=2^{(b+1)g-1}$. Under any assignment to these partial products, the strip's weighted sum can only decrease from one layer to the next, from carries going out of the strip's top column. The strip's weighted sum therefore stays below $2^{(b+1)g-1}$ in every layer, and hence every variable in the strip's top column $(b+1)g-1$ is $0$. In particular
every adder in the top column has all inputs $0$ and so every carry entering the next strip is $0$.
\end{proof}

\begin{lemma}[Carry-lookahead adder strip-locality]
\label{lem:cla-local}
Fix $g\ge1$ and let a carry-lookahead adder (Definition~\ref{def:cla}) add two bit-vectors $\vv{x},\vv{y}$ whose bits in the top
column of every $g$-strip are the constant $0$. Then every wire of the adder is constant or a Boolean function of the summand bits
in the columns of a single strip.
\end{lemma}

\begin{proof}
Let $p_i,q_i$ and $c_i$ denote the bitwise signals and carries of the carry-lookahead adder (Definition~\ref{def:cla}). At the top column $s$ of every strip we have
$x_s=y_s=0$, hence $p_s=q_s=0$ and $c_{s+1}=0$. It remains to show that the wires of the higher levels, which compute the
carries through group pairs rather than through the carry recurrence, also depend only on the columns of a single strip. Unwinding
the carry recurrence over a contiguous interval $I=[a,i]$ shows that its carry-out is $c_{i+1}=Q_I\oplus(P_I\wedge c_a)$, where
\[
P_I=\bigwedge_{t=a}^i p_t,
\qquad
Q_I=\bigoplus_{t=a}^i \Bigl(q_t\wedge\bigwedge_{u=t+1}^i p_u\Bigr).
\]
Suppose $I$ contains a top column, and let $s$ be the highest top column it contains. Then $p_s=q_s=0$, so $P_I=0$, and every term
of $Q_I$ with $t\le s$ vanishes, since it contains the factor $q_s$ if $t=s$ and the factor $p_s$ if $t<s$. Consequently
$(P_I,Q_I)$ depends only on the signal pairs $(p_t,q_t)$ with $t>s$, which lie in the strip containing column $i$.

We now prove the lemma for each type of wire in the carry-lookahead adder: the bitwise signals, group pairs, carries, and sum
bits. A bitwise signal $p_i$ or $q_i$ depends only on the summand bits of column $i$. A group pair $(p^\ell_a,q^\ell_a)$ equals
$(P_I,Q_I)$ for the interval $I=[a,a+4^\ell)$ spanned by its level-$\ell$ sub-adder, and is therefore constant or depends only on
the signal pairs of a single strip. A carry wire computes some carry $c_i$ of the recurrence, and taking $I=[0,i-1]$ with $c_0=0$
gives $c_i=Q_{[0,i-1]}$. Since $c_i=0$ whenever $i-1$ is a top column, we may assume that $i-1$ and $i$ lie in the same strip. Then, from the previous paragraph, $c_i$ is constant or depends only on the signal pairs of this strip. This implies that a sum bit
$(x+y)_i=p_i\oplus c_i$ also depends only on the summand bits of the strip containing $i$.
\end{proof}

\wallacestriplocal*

\begin{proof}
Every gate of the multiplier has constant fan-in, so the clauses of the encoding have constant width. To verify locality, fix
$g\ge1$ and a $g$-sparse restriction. By Lemma~\ref{lem:wallace-tree-local}, the bits of the two final rows in every strip's top
column are the constant $0$, so Lemma~\ref{lem:cla-local} applies to the final carry-lookahead adder. Together, the two lemmas show
that every wire of the multiplier is constant or a Boolean function of the partial products belonging to a single strip.
\end{proof}

\section{Lean formalization}
\label{app:lean}

We provide a Lean~4 formalization of the main theorems, generated by Claude Fable 5.1, at \url{https://github.com/dysfunctor/associativity}. We are not aware of an established library for proof complexity or arithmetic circuits, so this formalization includes an implementation of these definitions. While we have inspected these definitions ourselves, we caution that they are a potential source of error.

\section*{AI Methodology}

The main ideas and the core proof of the associativity lower bound are due to GPT-5.6-Sol. This materially affected Sections~\ref{sec:hard-source}, \ref{sec:rho}, \ref{sec:reduction}, and~\ref{sec:frege}, and Appendix~\ref{app:architectures}. During the verification process, the author reformulated and rewrote the definitions, lemma statements, and proofs with the assistance of Claude Fable 5, and added attributions along the way. The author wrote initial drafts of the expository text, including the introduction,
proof outline, and conclusion, and used Claude Fable 5 afterwards to improve presentation.

The author produced the diagrams with assistance from Claude Fable 5. The author drew each figure by hand, sent a photo to Claude, and then refined the generated
diagram with further prompting.

The Lean formalization (Appendix~\ref{app:lean}) was generated by Claude Fable 5.1.

The author verified the correctness and originality of all content, including references, and takes full responsibility for the paper's content.

\bibliographystyle{amsalpha}
\bibliography{references}

\end{document}